\documentclass[10pt,journal,compsoc]{IEEEtran}

\ifCLASSOPTIONcompsoc
\usepackage[nocompress]{cite}
\usepackage{graphicx}
\usepackage{amsmath}
\usepackage{amssymb}
\usepackage{amsthm}
\usepackage{subfigure}
\usepackage{epsfig}
\usepackage[ruled,linesnumbered]{algorithm2e}
\usepackage{color, soul}
\graphicspath{{eps/}}
\newtheorem{theorem}{Theorem}
\newtheorem{lemma}{Lemma}
\theoremstyle{definition}
\newtheorem{defn}{Definition}

\newcommand{\rone}[1]{\definecolor{green}{RGB}{255,255,255} \sethlcolor{green}\hl{#1}}
\newcommand{\rtwo}[1]{\definecolor{yellow}{RGB}{255,255,255}  \sethlcolor{yellow}\hl{#1}}
\newcommand{\rthree}[1]{\definecolor{red}{RGB}{255,255,255} \sethlcolor{red}\hl{#1}}

\else
\fi

\ifCLASSINFOpdf

\else

\fi

\begin{document}

\title{Failure Aware Semi-Centralized Virtual Network Embedding in Cloud Computing Fat-Tree Data Center Networks}

%

    \author{Chinmaya~Kumar~Dehury, ~\IEEEmembership{Member,~IEEE}, Prasan~Kumar~Sahoo, ~\IEEEmembership{Senior Member,~IEEE} 
    \thanks{Chinmaya Kumar Dehury was with the department of Computer Science and Information Engineering, Chang Gung University, Guishan, Taiwan.
    Currently, he is with the Mobile \& Cloud Lab, Institute of Computer Science, University of Tartu, Estonia. Email: chinmaya.dehury@ut.ee.}
    \thanks{Prasan Kumar Sahoo (Corresponding Author) is with Department of Computer Science and Information Engineering, Chang Gung University, Taiwan. He is an Adjunct Research Fellow in the Division of Colon and Rectal Surgery, Chang Gung Memorial Hospital, Linkou, Taiwan.
    Email: pksahoo@mail.cgu.edu.tw}
   }


\markboth{IEEE Transactions on Cloud Computing,~Vol.~xx, No.~x, month~year}%
{Shell \MakeLowercase{\textit{et al.}}: Bare Demo of IEEEtran.cls for Computer Society Journals}
%

\IEEEtitleabstractindextext{
\begin{abstract}
\rone{In Cloud Computing, the tenants opting for the
Infrastructure as a Service (IaaS) send the resource requirements
to the Cloud Service Provider (CSP) in the form of Virtual Network
(VN) consisting of a set of inter-connected Virtual Machines (VM).
Embedding the VN onto the existing physical network is known as
Virtual Network Embedding (VNE) problem. One of the major research
challenges is to allocate the physical resources such that the
failure of the physical resources would bring less impact onto the
users' service. Additionally, the major challenge is to handle the
embedding process of growing number of incoming users' VNs from
the algorithm design point-of-view. Considering both of the
above-mentioned research issues, a novel Failure aware
Semi-Centralized VNE (FSC-VNE) algorithm is proposed for the
Fat-Tree data center network with the goal to reduce the impact of
the resource failure onto the existing users. The impact of
failure of the Physical Machines (PMs), physical links and network
devices are taken into account while allocating the resources to
the users. The beauty of the proposed algorithm is that the VMs
are assigned to different PMs in a semi-centralized manner. In
other words, the embedding algorithm is executed by multiple
physical servers in order to concurrently embed the VMs of a VN
and reduces the embedding time. Extensive simulation results show
that the proposed algorithm can outperform over other VNE
algorithms.}


\end{abstract}

}

\maketitle

\IEEEdisplaynontitleabstractindextext

%
\IEEEpeerreviewmaketitle

\IEEEraisesectionheading{\section{Introduction}\label{sec:intro}}
\IEEEPARstart{V}{irtualization} is the key technology that enables
Infrastructure Provider (IP) to share the same physical servers or
Physical Machines (PMs) among multiple tenants by creating virtual
version of the resources such as storage, memory, CPU and network.
This allows the system developer to decouple the service instances
from the underline hardware resources by placing the
virtualization software between the user's operating system and
the underlined physical resources \cite{cerroni_2017}. On the
other hand, a large number of PMs in a data center is connected
through several types of switches and long cables. Fat-Tree
\cite{2017-1} network topology is widely used in commercial data
centers in order to establish the connection among PMs through
core switches, aggregation switches and edge switches as the basic
connecting devices. The resources of the PMs are provided to the
tenants by creating multiple Virtual Machines (VMs), ensuring the
basic privacy properties, better QoS and other user requirements
as mentioned at the time of Service Level Agreement (SLA)
establishment between the tenant and the IP \cite{2016-1}.



The network of multiple VMs that are requested by the cloud users
is called Virtual Network (VN). The job of the Cloud Service
Provider (CSP) is to provide a set of interconnected PMs onto
which the VN can be embedded, which is known as Virtual Network
Embedding (VNE) problem \cite{2016-ZehengYang}. This problem can also be relate to the problem of virtual network function placement \cite{infocom18-fei}.
The job of virtual network embedding is to embed VMs and virtual
links. VM embedding refers to as  selecting a suitable PM by taking
location, energy and users' budget constraint parameters into
consideration \cite{2017-5-Liang}. Virtual link embedding refers
to as finding the shortest path between the corresponding PMs
taking number of switches, load on the physical links, latency and
several other parameters into consideration.

%
The unexpected failure of resources is one of the major research issues that cannot
be ignored by the cloud service provider. Failure in resources has
a direct impact on the services. Service failures have direct
impact on QoS, energy consumption, SLA violation, and huge revenue
loss of CSP \cite{LI2018887}. One of the major research issues in VNE is
how to avoid the unexpected service failures due to the failure of
any hardware module or software module. The failure of resources
includes switch failure, physical link failure, and PM failure
\cite{NazariCheraghlou201681}. The major reasons behind the
resource failures could be the wrong input to the user program,
high workload on switches and PMs and failure of power supply.
It is reported that a single failure (or cloud outage) for one hour could bring more than \$336,000 financial loss to cloud providers, such as Amazon and Microsoft \cite{endo2017minimizing, cloudoutage2014}. The impact of the potential failure of any physical resource on service availability, productivity, business reputation, etc. \cite{Gunawi2016} must be analyzed before embedding the VN onto the physical network.
Here, we are considering the failure of PM, the physical network
link and the switches. The VMs of the incoming VN must be
distributed in such a way that impact of physical resource failure
is minimized by avoiding potential resource failure. \rthree{It is
also essential to design an algorithm that can embed the VNs
considering the failure impact of the physical resources. An
efficient VNE algorithm may consume significant time only for the
embedding purpose. The nature of the algorithm such as
centralized, distributed or semi-centralized also plays a
major role in making embedding process faster, resulting in better
QoS.}

%
In general, the VNE algorithm is installed in a dedicated server
to accomplish the embedding process with the responsibility to
embed all incoming VNs onto the physical network in a centralized
manner. The major flaws of the centralized embedding algorithms
are poor scalability, lower capital message, single point of
failure and higher embedding cost \cite{2015-1}. On the contrary,
distributed embedding algorithm reduces the embedding cost and
delay time, which occurs due to the arrival of a large number of
simultaneous VNs, but does not obtain the global view of the
physical network. \rthree{As distributed VNE algorithms are
scalable unlike the centralized one, they give better performance
as compared to the centralized approaches in the geo-distributed
Data Centers (DCs)}\cite{2015-1}. Despite several advantages of
the distributed VNE algorithms, obtaining an optimal embedding
solution without global view of the DC network incurs heavy
resource requirement and time consumption.
To overcome the disadvantages of the centralized and distributed
VNE algorithms, a semi-centralized algorithm may produce a
near-optimal embedding solution by considering the global view of
the DC network and also can embed multiple VNs in
parallel. The entire semi-centralized algorithm is divided into two sub-algorithms, which may run in two stages. If the centralized approach is followed by one sub-algorithm in the first stage, the distributed approach is followed by another sub-algorithm in the second stage. However, if the distributed approach is followed in the first stage, the centralized approach will be followed in the second stage.
Multiple constraints such as revenue maximization, guaranteeing
SLA with higher QoS and power consumption maximization need to be
taken into consideration while designing an efficient
semi-centralized VNE algorithm.

\subsection{Motivation}

As discussed before, VNE problem refers to as embedding of VN onto
a set of interconnected PMs. In the process of embedding, CSP
needs to consider multiple factors such as resource availability
of the PMs, bandwidth availability, utilization of the resources,
load balancing among all the PMs, energy consumption of the PMs,
network latency etc. Failure probability of the PMs is one of the
main factors, which plays a crucial role in the process of
embedding. Without considering the PMs' failure probability, CSP
may embed the VMs onto the PMs that are most likely to fail in the
next time instance. Failure of PMs results in failure of
corresponding VMs. As a result, the corresponding services would
be interrupted, which may bring a huge monetary loss for the CSP.
Though a number of research articles have been proposed recently
to address the VNE problems considering different factors, the
failure probability of the PMs and impact of resource failure onto
the existing users' request has not been studied extensively.

Fig. \ref{fig:motivation} illustrates the impact of resource
failure on the users' request. Let $VN1$ and $VN2$ be the two VNs
received by the CSP from two different users as shown in Fig.
\ref{fig:motivation}(a). $VN1$ consists of VMs $V1, V2,$ and $V3$.
$VN2$ consists of VMs $V4$ and $V5$. Considering the example given
in Fig. \ref{fig:motivation}, let the CSP consists of two PMs
$PM1$ and $PM2$ connected in Fat-Tree topology with number of
ports $k=2$. The two VNs $VN1$ and $VN2$ are embedded onto these
two PMs. VMs $V1, V2,$ and $V5$ from $VN1$ and $VN2$,
respectively, are embedded onto $PM1$. VMs $V3$ and $V4$ from VN
$VN1$ and $VN2$, respectively, are embedded onto $PM2$. In this
scenario, if physical machine $PM2$ is failed, it has direct
impact on both users and VM $V3$ and $V4$ will be affected, as
shown in Fig. \ref{fig:motivation}(b). Besides the impact on the
VMs, the virtual link between VM $V3$ and $V1$ and the virtual
link between VM $V4$ and $V5$ will be affected as the one of the
end VMs of those virtual links are inactive or not in running
state. As shown in Fig. \ref{fig:motivation}(c), the impact of the
failure of network resource onto the existing users is discussed.
Failure of a physical link between core switch $c$ and aggregation
switch $A2$ has also direct impact on all users, which can affect
all five virtual machines hosted by $PM1$ and $PM2$. Failure of
this physical link also affects the virtual links $(V1, V3)$ and
$(V4, V5)$ of $VN1$ and $VN2$, respectively. Though Fat-Tree
network topology  offers higher degree of fault tolerant service
through multiple physical paths between any two servers, however
in a much highly loaded network, failure of one link brings higher
impact on the corresponding servers and the VMs. It is assumed
that the virtual machines within a VN can be connected in any
topology.

\begin{figure}[t]
    \centering \epsfig{file=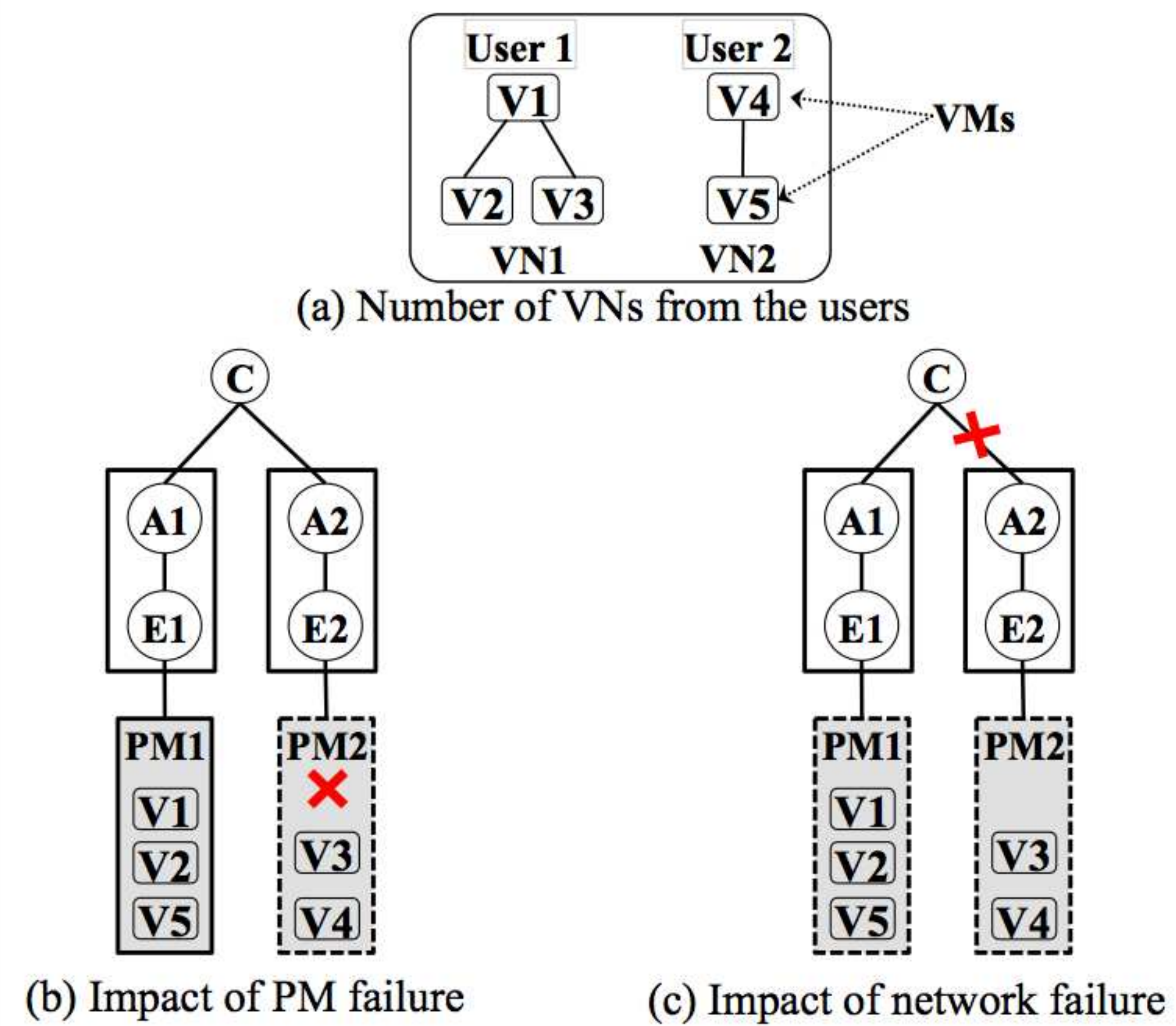,width=65mm}
    \vspace{-2mm}
    \caption{Impact of resource failure on VNE.} \vspace{-5mm}
    \label{fig:motivation}
\end{figure}

Further, several centralized VNE algorithms are not
scalable and are mainly designed to deploy on a single PM, which
is further responsible for receiving and processing the incoming
VNs \cite{LVRM-dehury, 2017-2}. This centralized approach
makes the algorithm in-scalable and therefore the new VNs may
experience delay as the single PM is busy in embedding other VNs
that arrived earlier. Here, in-scalable refers to as the lack of
ability to handle the growing number of incoming VNs to
accommodate itself according to a number of VNs. Besides, this
disadvantages in-terms of in-scalability and the knowledge of
global view of the entire physical network help the algorithm to
make the most efficient embedding solution in terms of the time
required to embed the VNs excluding the waiting time. The
distributed approach of VNE algorithm is highly scalable. Despite
its higher scalability nature, it is difficult to obtain an
efficient embedding solution as the global view is not available
for the PMs executing the embedding algorithm. Considering these
two traditional approaches and their disadvantages, it is
essential to design a VNE algorithm, which can be highly scalable
and must consider the global view of the entire physical network.

The above-mentioned scenario motivates us to revisit the VNE
problem by taking the failure of PMs as well as the physical links
into account. Accordingly, we have proposed a novel Failure aware
Semi-Centralized Virtual Network Embedding (FSC-VNE) scheme with
the following goals:
\begin{itemize}
    \item To design semi-centralized VNE
    algorithm that can meet the requirement of data centers with large volume of requests arrival rate.
    \item To minimize the impact of failure of PMs and physical links onto the existing users while embedding VNs by considering the failure probability of the PMs.
\end{itemize}

The goals can be achieved by embedding the VMs onto the PMs
considering the failure probability of each PM. Further, in order
to minimize the impact of physical link failure, multiple VMs can
be embedded onto a single PM, which eventually eliminates the
failure possibilities of corresponding virtual links.

Rest of this paper is organized as follows. In Section
\ref{sec:relWorks}, we present the recent research articles
related to VNE. The system model of this paper is presented in
Section \ref{sec:sysModel}. The proposed VNE algorithm is
discussed in Section \ref{sec:sol}. In Section
\ref{sec:perfEvaluation}, performance evaluation of our proposed
algorithm is made followed by the concluding remarks in Section
\ref{sec:conclsn}.

\section{Related Works} \label{sec:relWorks}
Extensive researches have been emphasized on different aspects of VNE problem, such as energy-aware \cite{IWQos16-Xu} 
location constraint \cite{2016-3}, physical network topology \cite{2016-1}, etc. Considering the location
preferences of the users, authors in \cite{2016-3}, formulate the
VNE problem as a graph bisection problem, which achieves integrated
VM and virtual link mapping. Restricting one PM to host
maximum one VM in order to reduce the failure impact onto single
VN arises another problem of accumulative failure impact onto all
corresponding VNs.
\subsection{Resource utilization optimization}
In order to reduce the total amount of required physical resources
in a multi-cloud environment, priority values are assigned to the
user's requests \cite{2017-2}. Here, each PM can be used to host a maximum of one VM from a VN. Such an approach contradicts the goal of minimizing the required physical resources. Assigning multiple VMs onto a single PM can
effectively reduce the required amount of resources. In a likewise
fashion, degree and clustering information of a VM within a VN are
used to embed the user's request onto physical network with the
objective to minimize the network utilization and maximize the
acceptance ratio as in \cite{2016-5}. However, the dependencies
among the VMs are not considered, resulting in inefficient mapping
solution.

In a multi-service cloud environment, resource allocation and
provisioning algorithms are presented in \cite{8027098}
considering the SLA, service cost, system capacity and resource
demand. However, considering only VM resources without network
resource may not fit into the real-life scenario. Authors in \cite{2016-7-Yin} propose a novel VNE algorithm
considering the link interference due to bandwidth
scarcity. The major advantage of the proposed I-VNE algorithm is
that the temporal and the spatial network topology information are
taken into consideration and the VNs are embedded with low
interference.

To efficiently utilize the computing resources and increase the
acceptance ratio, authors in \cite{2017-5-Liang} presented a
unique approach to embed the VN. This is done by simply placing the same VM onto multiple PMs and
distributing the input data onto corresponding PMs. However, this
approach required extra network resource and the number of PMs
required for a VN is more than the number of VMs present in the
VN. VNF placement problem is further studied in different scenarios, such as 5G network slicing framework \cite{8737660}, Network Function Virtualization (NFV) \cite{IWQos16-Xu}. Authors in \cite{xiao2019nfvdeep} address the Virtual Network Functions (VNFs) placement problem with Deep Reinforcement Learning approach to obtain a near-optimal solution. However, the proposed algorithm does not consider the impact of the failure of physical resources that may bring a more significant effect on the VNFs and the resource utilization.

In federated cloud, taking minimization of the network latency as
the major goal, authors in \cite{2016-4} propose network-aware VNE
scheme. The proposed algorithm ignores the fact that the single PM
can be used to execute multiple VMs, which eventually eliminate
the network resource demand without compromising the QoS and other
user's requirements. Similarly, the proposed congestion-aware VNE
scheme with Hose model abstraction in \cite{2017-congestion-VNE},
has the restriction to embed only one VM onto one PM. As a result,
the opportunity to minimize the amount of required network
resource is ignored. A rank-based VNE algorithm is presented in
\cite{8107491} considering substrate node degree, strength and
distance with others. However, its inability to scale and
inability to consider the failure probability of the substrate
nodes and links makes the VNE approach inefficient.
\subsection{Embedding time minimization}
Authors in \cite{2015-1} propose a
distributed embedding algorithm, which partitions the entire
physical network into multiple blocks. Though
the embedding cost for each VN is reduced, however, the
accumulative computational overhead still exists and the total
amount of required computational resources is very high. In order
to fasten the embedding process for the VN that follow mess
topology, authors in \cite{2016-8-Mano}, introduced a novel
approach that reduces the number of virtual links. However, combining
multiple virtual links onto one increases the embedding cost. This
trade-off is addressed by proposing reduction algorithms. However,
this centralized algorithm may not be able to scale up with the growing number of PMs and hence may degrade the efficiency in terms of embedding time.

The VNE problem can be viewed as multi-objective linear
programming problem, as in \cite{Wang2016196}. Though the authors
have achieved the goal to maximize the revenue and minimize the
embedding cost, however, the mapping solution may suffer from
physical resource fragmentation  issue and may not produce
efficient result in terms of embedding cost due to the mapping of
virtual node followed by the virtual links.
\subsection{Cost and Revenue}
Two VNE algorithms are presented based on Monte Carlo tree search and
multi-commodity flow algorithm in \cite{7859375} with the
objective to maximize the profit of the IP. The proposed centralized
VNE algorithm claims to map the virtual nodes and the virtual
links without splitting the physical path. However, the
centralized approach may produce an inefficient result in terms of
embedding time. In \cite{HESSELBACH201614}, a new path algebra
strategy is proposed to embed the
virtual nodes and the virtual links based on cost and revenue parameter.
However, the proposed multi-staged mapping strategy may produce
extra computational overhead, resulting higher embedding time.

Authors in \cite{8046038} address the problem of
inefficient bandwidth reservation due to the uncertainty. Though
the proposed stochastic approach achieves maximizing the revenue
of the CSP, however, author should also consider the failure
probability of the physical links.

A number of VNE approaches have been proposed considering dynamic
resource requirement, topology of the data center, energy
consumption, and revenue maximization. Though the studies in current
literature solve many real-life scenarios in cloud, to the best of
our knowledge none of the studies propose an embedding scheme,
which processes the VNs in semi-centralized manner to fasten the
embedding process and minimizes failure impact of the physical
resources onto the existing VNs.

\section{Problem Formulation} \label{sec:sysModel}

As discussed earlier, it is the job of the CSP to create/embed
multiple interconnected VMs onto a set of suitable PMs. In the
embedding process, the network topology plays an important role as
it defines the network structure and the connectivity of one PM
with other PMs. In our study, Fat-Tree network topology is taken
into consideration as the data center network topology. It is
assumed that the computing and network resources of PMs and the
networking devices remain static throughout the execution of the
algorithm for any particular VN.\rthree{In other words, the
resource capacities of the PMs and networking devices do not
change, when the VNE allocation algorithm runs.} The PMs are
heterogeneous in nature.
The resource requirement of the VN includes the computing and
network resource demand of VMs and virtual links, respectively,
which do not change over time. Further, the  topology of the VN
including the number of VMs, virtual links, the resource demand of the VMs and the virtual links are static and
hence do not change after the VN request is submitted to the CSP.
In further sections, we will discuss the network topology, the VN,
and failure impact of physical resources onto the VNs in details. A summary of all the notations is given in Table \ref{table:notations}.

\begin{table}
   \centering
   \caption{List of notations}
   \begin{tabular}{ |c|p{6.5cm}| }
       \hline
       \textbf{Notation} & \textbf{Description}\\
       \hline
       $P$ & Set of $m$ PMs $\{p_1, p_2, \dots, p_m\}$\\ \hline
       $S$ & Set of $s$ switches $\{S_1, S_2, \dots, S_s\}$\\ \hline
       $k$ & number of ports available at each switch\\ \hline
       $V$ & Set of $n$ VMs $\{v_1, v_2, \dots, v_n\}$ \\ \hline
       $x$ & resource type $\{memory, CPU\}$ \\ \hline
       $\alpha_i^x$ & Server resource requirements of VM $v_i$ \\ \hline
       $\alpha_{ij}^e$ & Bandwidth requirement between VM $v_i$ and $v_j$\\ \hline
       $e_{ij}$ & The virtual link between VM $v_i$ and $v_j$\\ \hline
       $\beta_i^x$ & Resource of type $x$ available at PM $p_i$\\ \hline
       $\beta_i^n$ & Network resource available at PM $p_i$\\ \hline
       $L_i^j$ & $1: $if there exist a link between PM $p_i$ and switch $S_j$\\ \hline
       $L_{ij}$ & $1: $if there exist a link between switch $S_i$ and switch $S_j$ \\ \hline
       $\omega_{ij}$ & Number of switches present between PM $p_i$ and PM $p_j$ \\ \hline
       $\psi_i$ & Workload on PM $p_i$\\ \hline
       $\hat{\alpha}_i^e$ & Bandwidth requirement by VM $v_i$\\ \hline
       $I_s^v(p_i)$ & The impact of failure of PM $p_i$ onto the set of VMs\\ \hline
       $\kappa_i^j$& $1: $if VM $v_j$ is assigned to PM $p_i$\\ \hline
       $\lambda_{ij}^{qr}$ & $1: $ if the virtual link $e_{qr}$ is assigned to physical link $L_ij$ or $L_i^j$\\ \hline
       $I_N^n(z_{ij})$ & Impact of physical link $z_{ij}$ failure on the set of virtual links\\ \hline
       $I_N^v(z_{ij})$ & Impact of physical link $z_{ij}$ failure on the set of VMs\\ \hline
       $I_w^v(s_j)$ & Impact of network switch, $s_j$, failure onto the set of VMs\\ \hline
       $I_w^n(s_j)$ & Impact of network switch, $s_j$, failure onto the set of virtual links\\ \hline
       $PE^e(p_i)$ & Set of edge layer neighboring PMs of the PM $p_i$ \\ \hline
       $PE^a(p_i)$ & Set of aggregation layer neighboring PMs of the PM $p_i$ \\ \hline
       $PE^c(p_i)$ & Set of core layer neighboring PMs of the PM $p_i$ \\ \hline
   \end{tabular}
\vspace{-2mm}
\label{table:notations}
\end{table}


\subsection{Fat-Tree data center network}

As shown in Fig. \ref{fig:example_fat-tree_vn}(a), the switches are mainly organized into three layers, such as core switch, aggregation switch, and edge switch in the Fat-Tree data center network \cite{2017-2}.
The number of layers in Fat-Tree
can be extended to more than three layers. However, we focus our
study considering the basic three-layer Fat-Tree network topology.
Let $k$ be the number of ports available at each switch. Being at
the root of the network, the number of core switches are
$(\frac{k}{2})^2$. The total number of aggregation switches is equal to
the total number of edge switches required, which is
$\frac{k}{2}$. The maximum number of servers that can be
attached to the Fat-Tree topology is $\frac{k^3}{4}$. Let
$S=\{S_1, S_2, S_3, ..., S_s\}$ be the set of $s$ number of
switches, which includes core switches, aggregation switches, and
edge switches. The value of $s$ can be calculated as
$s=\frac{5k^2}{4}$.

\begin{figure}[h]
    \begin{center}
    \epsfig{file=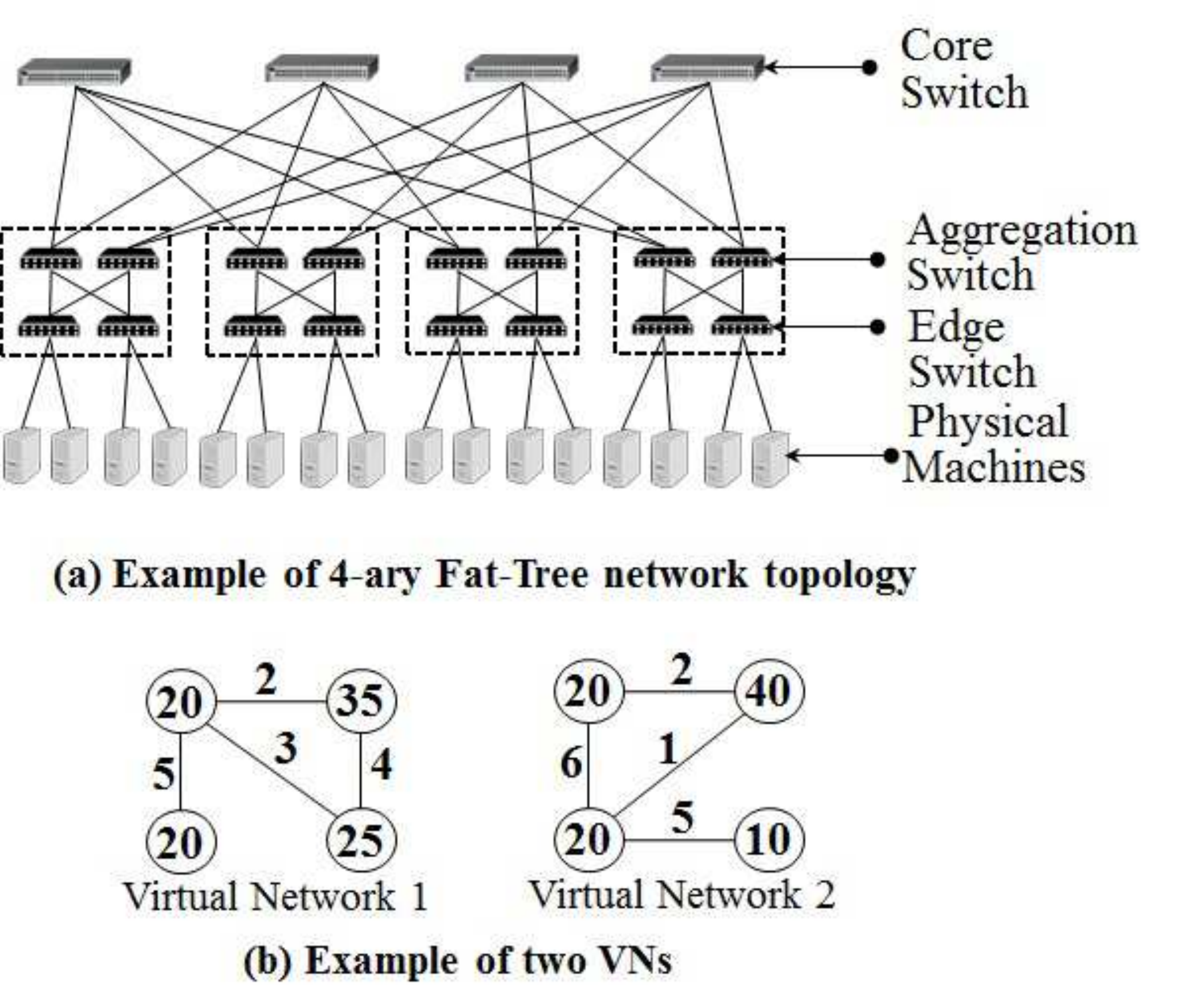,width=65mm}
    \vspace{-2mm}
    \caption{Example of Fat-Tree network topology and virtual network.}
    \vspace{-5mm}
    \label{fig:example_fat-tree_vn}
    \end{center}
\end{figure}

Let $P=\{p_1, p_2, p_3, ..., p_m\}$ be the set of $m=\frac{k^3}{4}$ number of PMs that are attached to the bottom level of Fat-Tree topology.
Let $\beta_i^x$ be the resource of type $x \in \{memory, CPU\}$
available at PM $p_i$. For instance, $\beta_2^{memory}$ represents
the amount of remaining memory available at PM $p_2$. $L_i^j$ is
the boolean variable, which represents if there exists a physical
link between switch $S_i, 1 \le i \le s$ and PM $p_j, 1 \le j \le
m$. Mathematically,
\begin{equation}
L_i^j = \left\lbrace
\begin{array}{cc}
1 & \text{if physical link between PM } p_i \\
  & \text{ and Switch }S_j\text{ exist}\\
0 & \text{Otherwise}
\end{array} \right.
\end{equation}
 Similarly, $L_{ij}$ is the boolean variable, which indicates if
 there exists a physical link between switch $S_i, 1 \le i \le s$ and $S_j, 1 \le j \le s, i \ne j$. Mathematically,
 \begin{equation}
 L_{ij} = \left\lbrace
 \begin{array}{cc}
 1 & \text{if physical link between switch } \\
   & S_i\text{ and }S_j\text{ exist}\\
 0 & \text{Otherwise}
 \end{array} \right.
 \end{equation}

In the three-layer Fat-Tree network topology, the number of links
from any PM to any one of the core switches is three. They are (1)
the link from the PM to the edge switch, (2) the link from the
edge switch to the corresponding aggregation switch, and (3) the
link from the aggregation switch to the core switch. The maximum
bandwidth capacity of all the links is same according to the
characteristics of the Fat-Tree network topology \cite{2017-1}.
The physical links at the lower level of the tree structure
network topology usually carry less workload as compared to the
links at the upper level. The $k$ ports in a switch are divided
into two groups of $\frac{k}{2}$ ports. One group is used to
connect to the upper layer switches and another group of ports is
used to connect to the lower layer switches. This infers the fact
that the failure of physical links at the lower level has less
impact than the failure of the upper level physical links.

\subsection{Virtual Network}\label{sec:sysModel:VN}
As depicted in Fig. \ref{fig:example_fat-tree_vn}(b), the VN is
represented as undirected weighted graph, where the node
represents the VM and the edge (or virtual link) represents the
communication between corresponding pair of VMs. Let, $V=\{v_1,
v_2, v_3, ... , v_n\}$ be the set of $n$ number of VMs. The
resource requirement of the VM $v_i, 1\le i \le n$ is represented
as $\alpha_i^x$, where $x \in \{memory, CPU\}$ defines the
resource type. Let, $e_{ij}$ and $\alpha_{ij}^e$ represent the
virtual link and the bandwidth requirement between VM $v_i$ and
$v_j, \quad i \ne j$, respectively. In the example given in Fig.
\ref{fig:example_fat-tree_vn}(b), the numeric value within each
circle represents the amount of required computing resource and
the value beside each edge represents the required network
bandwidth between the corresponding VMs.

\subsection{Definitions}

\begin{figure}[t]
    \begin{center}
        \epsfig{file=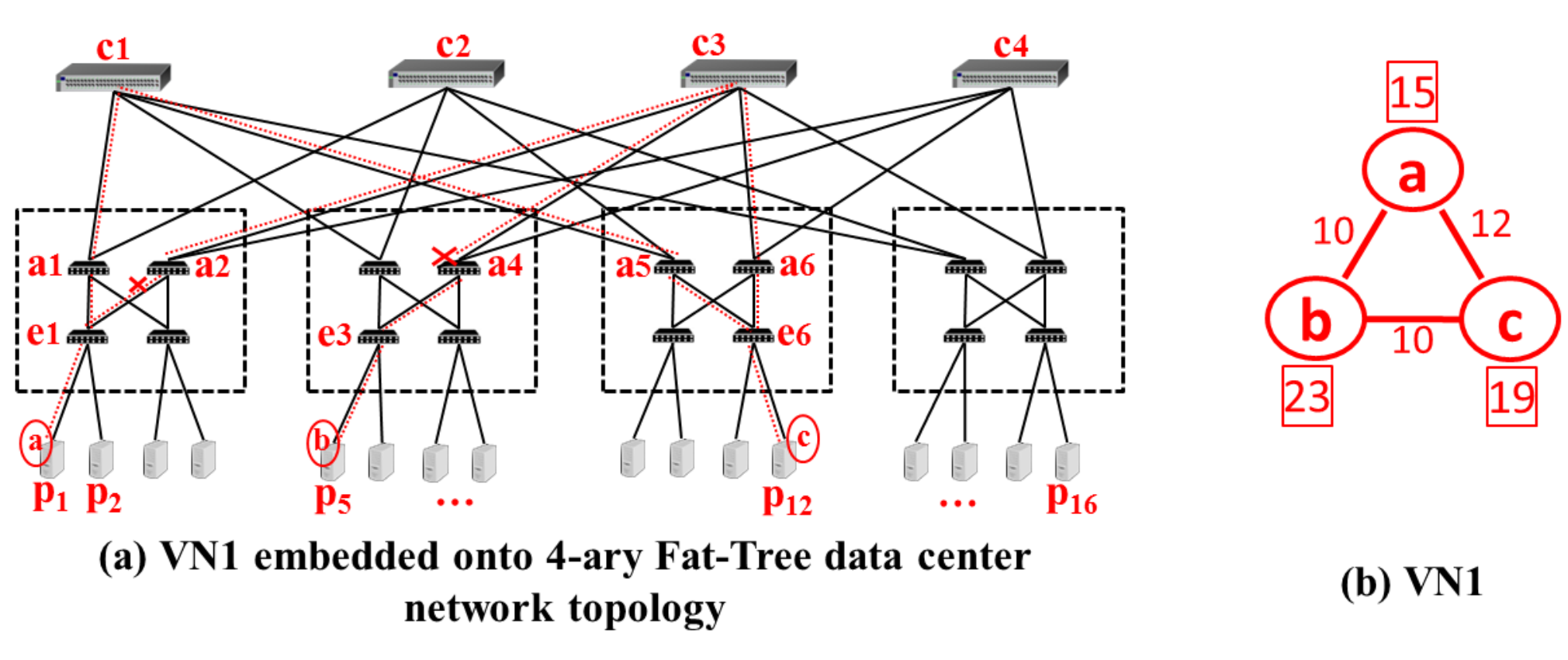,width=65mm}
        \vspace{-3mm}
        \caption{Resource failure impact onto virtual network.}
        \vspace{-6mm}
        \label{fig:resrc_failure_impact}
    \end{center}
\end{figure}

\begin{defn}{\textbf{Remaining Bandwidth at PMs:}}
The remaining available bandwidth at each PM is calculated by
finding the minimum remaining bandwidth available along each path,
which is originated from the corresponding PM to all core
switches. Mathematically,
    \begin{equation}\label{eq:def:pmReqBW}
    \beta_a^n = \min_{\forall S_i, S_j, S_l \in S} \left\lbrace
     \begin{array}{c}
     \gamma^r(p_a, S_i)*L_a^i, \\
     \gamma^r_a(S_i, S_j)*L_{ij},\\
     \gamma^r_a(S_j, S_l)*L_{jl}
     \end{array}\right\rbrace
    \end{equation}

Where, $\beta_a^n$ represents the remaining bandwidth available at
PM $p_a$. The first term $\gamma^r(p_a, S_i)*L_a^i$ represents the
remaining bandwidth available in the physical link between PM
$p_a$ and the connected edge switch. The middle term
$\gamma^r_a(S_i, S_j)*L_{ij}$ represents the remaining bandwidth
available at all physical links between the edge layer switches
and the aggregation layer switches that are connected to the same
pod, where PM $p_a$ is attached. Similarly, $\gamma^r_a(S_j,
S_l)*L_{jl}$ represents the remaining bandwidth available at all
physical links between the aggregation switches related to PM $p_a$ and
all core switches. For better understanding, let us consider the
example given in Fig. \ref{fig:resrc_failure_impact}(a). In order
to calculate the remaining bandwidth at PM $p_1$, we calculate the
remaining bandwidth available in all three layers of physical
links. Firstly, the link between PM $p_1$ and the edge switch
$e1$, secondly, the physical link between edge switch $e1$ and
aggregation switches $a1$ and $a2$, thirdly, the physical links
between the aggregation switches $a1$ and $a2$ and all core
switches, i.e. physical links $(a1,c1), (a1,c2), (a2,c3),$ and
$(a2,c4)$. The minimum bandwidth available among all
above-mentioned physical links is the remaining bandwidth
available at PM $p_1$.

\end{defn}

\begin{defn}{\textbf{Workload on PMs:}}\label{eq:def:pmWorkload}
    The workload on the PM $p_i$, represented as
    $\Psi_i$, can be calculated as the ratio between the amount of the
    physical resources allocated to other VNs and its maximum resource
    capacity.
\end{defn}

\begin{defn}{\textbf{Required Bandwidth of VMs:}}
    The bandwidth requirement by any virtual machine can be calculated by summing up the bandwidth requirement by the adjacent virtual links. Mathematically,
    \begin{equation} \label{eq:def:vmBWreq}
    \hat{\alpha}_i^e = \sum_{j=1}^{n}{\alpha_{ij}^e}
    \end{equation}

For instance, the total bandwidth required by the VM $a$ in Fig.
\ref{fig:resrc_failure_impact}(b) is calculated as the bandwidth
demand by the virtual link $(a,b)$ and $(a,c)$, i.e. $10+12=22$.
Similarly, the total required bandwidth by the VM $b$ and $c$ is
$20$ and $22$, respectively. In the proposed embedding scheme, it is assumed that the required bandwidth of all the virtual links are time-invariant and don't change over time.

\end{defn}

\begin{defn}{\textbf{Semi-Centralized VNE:}}
An embedding algorithm is said to be semi-centralized if the
algorithm follows both centralized and distributed approach for
embedding the VN onto the physical network. The first stage of the
proposed VNE algorithm follows centralized approach, where a set
of PMs is selected by a single embedder, whereas the second stage
of the proposed algorithm follows the distributed approach, where
the set of selected PMs executes the same algorithm to pick-out
suitable VM. Since the proposed embedding process follows both
centralized and distributed approach, it is called
semi-centralized VNE algorithm.
\end{defn}

\subsection{Resource failure impact} \label{sec:sysModel:resFailImpact}
We are mainly considering three types of resource failures, such
as (a) server or PM resource failure, which includes the CPU,
memory, storage, and other physical resources (b) physical link
failure, and (c) switch failure. Further, the resource failure
impacts are of two types, such as (a) the impact on the VMs, and
(b) the impact on the virtual links. In our study, the impact of
the failure of physical network resource and switch resource onto
both VMs and virtual links are taken into consideration. The impact of failure of PM resources onto the VMs is only taken into account. But, the impact of such failure onto the adjacent virtual links is not considered. However, the impact of server failure onto the virtual links
excluding the adjacent virtual links is not taken into account.
\rone{Here, it is assumed that the failure probability of the
servers can be obtained by the CSP by observing the historical
information of each server} \cite{2017-serviceRelibility}. The
failure probability of the physical machines may increase with the
increase in their workload,  users' interaction etc. The same
applies to the physical links also. The more communication and the
workload on links may lead to their higher failure probability.
Hence it is necessary to consider the failure probability of the
physical machines and the physical links, which can minimize the
failure impact onto the users' VNs.

\subsubsection{PM failure impact}
Let, $I_s^v(p_i)$ be the impact of failure of PM $p_i$ onto the
set of VMs $V = \{v_1, v_2,\dots,v_n\}$. The PM failure impact is
quantified as the value between 0 and 1, and is calculated as the
ratio of amount of allocated server resources affected and the
total amount of required resource. Mathematically, the impact of
physical server resource failure onto the virtual machine is

\begin{equation} \label{eq:PMFailImpactOnVMs}
I_s^v(p_i) = \frac{1}{|x|}*\sum_{\forall x}\frac{\sum_{j=1}^{n}\kappa_i^j * \alpha_j^{x}}{\sum_{j=1}^{n}\alpha_j^{x}}, \quad x \in \{memory, CPU\}
\end{equation}

The numerator in Eq.~(\ref{eq:PMFailImpactOnVMs}) represents the
amount of resource of type $x, x\in \{memory, CPU\}$ of the PM
$p_i$ is allocated to the VMs of the current VN. The denominator
represents the total amount of resource of type $x, x\in \{memory,
CPU\}$ is demanded by all VMs. $\kappa_i^j$ is the boolean
variable that indicates if the virtual machine $v_j$ is assigned
to the PM $p_i$.

\begin{equation} \label{eq:kappaDef}
    \kappa_i^j = \left\lbrace
    \begin{array}{cc}
     1 & \text{if VM }v_j\text{ is assigned to PM }p_i \\
     0 & \text{Otherwise}
    \end{array}
    \right.
\end{equation}
It is to be noted that equal importance is given to both types of
the resources. For better understanding, Let us assume an example
of a virtual network consisting of three VMs as \textit{a},
\textit{b}, and \textit{c} with memory and CPU resource demand of
$(10, 2), (23,3),$ and $(15,1)$ units, respectively. Let, VM
\textit{a, b, }and \textit{c} be deployed on PM \textit{A, B,} and
\textit{C}, respectively. The total memory resource demand of all
the three VMs is $10+23+15 = 48$. Similarly, $2+3+1=6$ is the
total CPU demand. Here, PM \textit{A} provides $10~units$ and
$2~units$ of memory and CPU resources, respectively to the VM
\textit{a}. The failure impact of PM \textit{A} onto the virtual
network can be calculated as
$\frac{1}{2}*\frac{10}{48}+\frac{2}{6} = 0.27$. Similarly, the
failure impact of PM \textit{B} and \textit{C} onto all the VMs
can be calculated as $0.49$ and $0.24$, respectively.

\subsubsection{Physical link failure impact}
In this subsection, we discuss the impact of physical link
failures on virtual links and the VMs.

\textit{Physical link failure impact on virtual links: }
$I_N^v$ and $I_N^n$ be the impact of physical link failure onto
the set of VMs and the set of virtual links, respectively. The
impact of the physical link failure onto the virtual links is
calculated as follows.

\begin{equation}\label{eq:NWfailImpactOnVirtualLink}
I_N^n(z_{ij}) = \frac{\sum_{q=1}^{n-1}\sum_{r=q}^{n}\lambda_{ij}^{qr} * e_{qr} * \alpha_{qr}^e}{\sum_{q=1}^{n-1}\sum_{r=q}^{n} e_{qr} * \alpha_{qr}^e }, \quad \forall i,j, i\ne j
\end{equation}

Where, the boolean variable $\lambda_{ij}^{qr}$ indicates if the
virtual link $e_{qr}$ is embedded onto the physical link $z_{ij}$,
which represents the physical link between two switches $L_{ij}$,
or the physical link between one switch and one PM $L_{i}^j$.
$\lambda_{ij}^{qr}$ can be expressed as

\begin{equation}\label{eq:lambdaDef}
    \lambda_{ij}^{qr} = \left\{
\begin{array}{cc}
1 & \text{if the virtual link } e_{qr} \\
  & \text{ is assigned to }  \text{physical link }  z_{ij} \\
0 & \text{Otherwise}
\end{array}\right.
\end{equation}

The numerator and denominator of the
Eq.~(\ref{eq:NWfailImpactOnVirtualLink}) represent the amount
of physical network resource allocated to the VN and the total
amount of network resource demanded by all virtual links,
respectively.

For the sake of clarity, an example is presented in Fig.
\ref{fig:resrc_failure_impact}. Assume VN1 contains three VMs $a,
b,$ and $c$ with the computing resource demand of $15$ units, $23$
units and $19$ units, respectively. The network bandwidth demands
of the virtual links $(a, b)$, $(b, c)$, and $(a, c)$ are $10$
units, $10$ units, and $12$ units, respectively. Let, the entire
physical network consists of 16 PMs. VM $a, b,$ and $c$ are
assigned to PM $p_1, p_5,$ and $p_{12}$, respectively, as shown in
Fig. \ref{fig:resrc_failure_impact}. The dotted line in Fig.
\ref{fig:resrc_failure_impact}(a) represents the assigned virtual
links. For example, the virtual link $(a, c)$ is assigned to the
physical path $(p_1, e1, a1, c1, a5, e6, p_{12})$. The virtual
link $(b, c)$ is assigned to physical path $(p_5, e3, a4, c3, a6,
e6, p_{12})$. Similarly, the virtual link $(a,b)$ is assigned to
the physical path $(p_1, e1, a2, c3, a4, e3, p_5)$. Under this
virtual network embedding, the impact of failure of the physical
link $(e1, a2)$ results in failure of communication between VM $a$
and $b$. Using Eq.~(\ref{eq:NWfailImpactOnVirtualLink}), the
impact can be calculated as $\frac{10}{10+10+12} = 0.31$.

\textit{Physical link failure on a single VM:} The failures of
the physical links have also impacted on the virtual machines as the
failures of the virtual links lead to no or delay communication
among VMs. \rone{The communication among the VMs within a VN is
required to finish the assigned jobs and therefore the failure of
a single virtual link has the impact on the corresponding VMs and entire VNs} \cite{2017-serviceRelibility, LVRM-dehury}.
Following equation is used to calculate the impact of physical
link failure onto the single VM.


\begin{equation} \label{eq:NWfailImpactOnVM1}
\begin{split}
I_N^v(z_{ij},v_q) & =\frac{\alpha_{qr}^{e}*\left(\frac{\alpha_q^x}{\sum_{r=1}^{n}\alpha_{qr}^e}\right)}{\alpha_q^x}\\
& =\frac{\alpha_{qr}^{e}*\alpha_q^x}{\alpha_q^x * \sum_{r=1}^{n}\alpha_{qr}^e} 
=\frac{\alpha_{qr}^e}{\sum_{r=1}^{n}\alpha_{qr}^e}, \quad r \ne q
\end{split}
\end{equation}

The notation $I_N^v(z_{ij}, v_q)$ represents the impact of failure
of physical link $z_{ij}$ onto the VM $v_q$. The notation $z_{ij}$
is used to represent the physical link $L_{ij}$ or $L_{i}^j$. In
the example given in Fig. \ref{fig:resrc_failure_impact}, the
impact of failure of physical link $(e1, a2)$ onto the VM $a$ can
be calculated as $\frac{10}{10+12}=0.45$.

\textit{Physical link failure on all VMs:} Extending further, the
impact of failure of physical link $z_{ij}$ onto all VMs can be
calculated as

\begin{equation}\label{eq:NWfailImpactOnVN}
I_N^v(z_{ij}) = \sum_{q=1}^{n-1}\sum_{r=q}^{n}{\lambda_{ij}^{qr}}\left[\frac{\alpha_{qr}^e\left( \frac{\alpha_q^x}{\sum_{u=1}^{n}{\alpha_{qu}^e}} + \frac{\alpha_r^x}{\sum_{u=1}^{n}{\alpha_{ur}^e}} \right)}{\sum_{u=1}^{n}{\alpha_u^x}} \right],
\end{equation}
where, the numerator represents the sum of amount of computing
resources of two corresponding VMs of each affected virtual link
and the denominator represents the total amount of the resource
demanded by all VMs. Applying Eq. (\ref{eq:NWfailImpactOnVN}) onto
the example given in Fig. \ref{fig:resrc_failure_impact}, the
impact of failure of physical link $(e1, a2)$ onto all VMs can be
calculated as $\frac{10}{10+12} + \frac{10}{10+10} = 0.95$.

\subsubsection{Switch failure impact}
$I_w^v$ and $I_w^n$ be the impact of network switch failure onto the VM and the virtual link, respectively.

The failure of the switches can be considered as the collective
failure of the attached physical links.
The impact of such failure onto
the allocated computing resources or the VMs can be calculated by
using the Eq.~(\ref{eq:NWfailImpactOnVN}) as follows

\begin{equation}\label{eq:SWfailImpactOnVMs}
I_w^v (s_j) = \sum_{i=1}^{s}{z_{ij}*I_N^v(z_{ij})}, \quad i \ne j
\end{equation}

Similarly, the impact of switch failure onto the allocated network
resources or the virtual links is calculated as follows.

\begin{equation}\label{eq:SWfailImpactOnNW}
I_w^n (s_j) = \sum_{i=1}^{s}{z_{ij}*I_N^n(z_{ij})} , \quad i \ne j
\end{equation}

Considering the example given in Fig. \ref{fig:resrc_failure_impact}, failure of the aggregation switch
$a4$ results in failure of virtual links $(a,b)$ and $(b,c)$.
However, failure of the core switch $c1$ results in failure of the
virtual link $(a,c)$.

\subsection{Neighbor PMs}
Neighbor PMs are used to find the set PMs that are nearer to each
other and can be used to embed a VN. As a result, the VMs those
are belonged to the single VN are placed as nearer as possible.
Neighbor PMs are of three types, such as (a) Edge layer neighbor
PMs (b) Aggregation layer neighbor PMs, and (c) Core layer
neighbor PMs. These are classified based on the number of
intermediate switches.

\subsubsection{Edge layer neighbor PM} Two PMs are said to be edge layer neighbor PMs, if both PMs are connected to the same edge switch. Mathematically,

\begin{equation} \label{eq:edgeLayerNeighborPM}
    {PE}^e(p_i) = \{p_j|w_i^j=1\}, 1 \le i \le m, 1 \le j \le m, i \ne j
\end{equation}
where, ${PE}^e(p_i)$ is the set of edge layer of neighboring PMs
of the physical machine $p_i$. $w_i^j$ is the minimum number of
intermediate switches present in between the PM $p_i$ and $p_j$.

\subsubsection{Aggregation layer neighbor PM} Two PMs are said to be aggregation
layer neighbor PMs, if both PMs are attached to the same pod. In
other words, the number of intermediate switches between both PMs
is $3$. Mathematically,

\begin{equation}\label{eq:aggrLayerNeighborPM}
PE^a(p_i) = \{p_j|w_i^j=3\}, 1 \le i \le m, 1 \le j \le m, i \ne j
\end{equation}

where, $PE^a(p_i)$ is the set of aggregation layer neighboring PMs
of physical machine $p_i$.

\subsubsection{Core layer neighbor PM}
Similar to edge layer and aggregation layer neighbors, the number of intermediate
switches among core layer neighbor PMs is $5$. This can be
represented mathematically as follows.

\begin{equation}\label{eq:coreLayerNeighborPM}
PE^c(p_i) = \{p_j|w_i^j=5\}, 1 \le i \le m, 1 \le j \le m, i \ne j
\end{equation}

\subsection{Objective function}
In the proposed failure-aware VNE scheme, the failure
probabilities of the physical machines are taken into
consideration in order to reduce the failure impact on the users'
VNs. \rthree{The failure impact of the physical machines, physical
links and switches are considered as calculated in Section
}\ref{sec:sysModel:resFailImpact}.\rthree{The failure impact onto
the VMs and virtual links are calculated in a balanced model
presuming that all the VMs and virtual links are not
associated with any additional parameter to decide their importance.
However, the VM or virtual link having higher resource demand has
direct influence over the failure impact. In other word, VM having
higher resource has maximum failure impact and similarly, the
virtual link having higher network resource demand has maximum
failure impact. }The main goal of the proposed work is to reduce
the failure impact onto the users' VN. This is done by embedding
the VMs onto the PMs in close proximity as discussed above. As a
result, along with the overall reduced failure impact, the
proposed VNE scheme would be able to minimize the total amount of
physical resources. Based on the above problem description, the
objective function of the proposed scheme is formulated as
follows:

\noindent \textbf{\rtwo{Objective:}}

\begin{align} \label{eq:objective}
\min  \quad \sum_{i=1}^m \left[ \varepsilon_1 * I_s^v(p_i)
+ \varepsilon_2 * \sum_{\forall r} \frac{I_N^v(Z_{ir})+I_N^n(Z_ir)}{2} \right.\nonumber \\
\left. + \varepsilon_3 * \sum_{\forall q} \frac{I_w^v(S_q)+I_w^n(S_q)}{2} \right]
\end{align}

\noindent \textbf{Constraints:}

\begin{equation} \label{const:compResrc}
\alpha_{j}^x \le \beta_i^x, \quad \forall v_j \in V, \forall p_i \in P, x \in \{CPU, Memory\}
\end{equation}

\begin{equation}\label{const:NWResrc} 
    \hat{\alpha}_j \le \beta_i^n, \quad \forall v_j \in V, \forall p_i \in P
\end{equation}

\begin{equation}\label{const:oneVM_onePM} 
    \sum_{i=1}^{m} {\kappa_i^j = 1}, \quad \forall v_j \in V
\end{equation}

\begin{equation} \label{const:weightFactor}
\varepsilon_1 + \varepsilon_2 + \varepsilon_3 = 1, \quad 0 < \varepsilon_1, \varepsilon_2, \varepsilon_3 < 1
\end{equation}


The objective function in Eq. \ref{eq:objective} is of
two-fold. The first part focuses on selecting the PM such that its
failure would bring minimum overall impact onto the assigned VN.
The second part focuses on selecting the adjacent physical links
and the switches such that their failure would bring minimum
impact onto the virtual machines and the virtual links. \rtwo{The
weight of the impact of the resources can be decided by modifying
the values of $\varepsilon_1, \varepsilon_1, $ and
$\varepsilon_3$. This would allow the CSP to give importance to
different resources at different time instances. For example,
higher value of $\varepsilon_2$ gives maximum importance to the
impact of the physical link failure over the failure of the
physical machines and switches. Similarly, equal of all
the weight factors give equal importance to the impact of failure
of the PMs, physical links, and switches.}

While evaluating the objective function, the proposed algorithm
must ensure that following constraints are satisfied. Constraint
(\ref{const:compResrc}) ensures that the available computing
resource of the selected physical machine is more than the
computing resource demand of the virtual machine. Similarly,
Constraint (\ref{const:NWResrc}) ensures that the total required
network bandwidth of a VM is less than the available bandwidth of
the PM. Constraint (\ref{const:oneVM_onePM}) ensures that no VM is
embedded onto multiple PMs. In the proposed VNE scheme, we take
the advantage of embedding multiple VMs of single VN onto a single
PM. This will allow us to ignore the corresponding virtual link
from embedding onto the physical network. \rtwo{Constraint
}(\ref{const:weightFactor}) restricts the values of
$\varepsilon_1, \varepsilon_1, $ and $\varepsilon_3$ to be less
than $1$ and is greater than $0$. However, the sum of all the
weight factors must be equal to $1$.

\section{Proposed VNE algorithm}\label{sec:sol}
\rtwo{In the above section, we have formulated the VNE problem and
have presented the objective function with the aim to minimize the
impact of the resource failure onto the virtual networks. In order
to meet the objectives mentioned in Eq.}
\ref{eq:objective}, we propose a novel heuristic two-staged
Failure aware Semi-Centralized Virtual Network Embedding (FSC-VNE)
algorithm, which is specially designed for a large number of
physical servers connected using the Fat-Tree network topology in a
data center. With one VNE request as the input, the entire algorithm works in two stages. First is
the PM Cluster Formation (PM-CF) stage, and second is the VM
Pick-out (VM-Po) stage. The proposed FSC-VNE algorithm follows the
semi-centralized approach for embedding the VN. As discussed in
Section~\ref{sec:intro}, the VNs are embedded by the single
dedicated server onto thousands of physical servers in the
first-come-first-serve manner in a centralized approach, which
introduces a waiting time for the VNs those arrive lately.
Further, in distributed approach of the embedding process, the VNs
are embedded by multiple dedicated servers concurrently, which
leaves the synchronization and communication overhead for the
physical servers as a major issue. Taking the advantages of both
approaches, the PM-CF stage is designed to preprocess the VNs in a
centralized manner, whereas the VM-Po stage is designed to work in
distributed approach with negligible communication overhead, as
described in Lemma \ref{lemma:CommOverhead}. It is to be noted
that, reducing the embedding time by distributing the incoming
virtual networks onto multiple servers is one of our goals.
However, being one of the major disadvantages of centralized
approach, single-point-failure  problem of the proposed FSC-VNE
scheme still exist, which is an obvious pitfall of both
centralized and semi-centralized algorithms.  As a whole, the
proposed FSC-VNE algorithm is said to be semi-centralized
embedding algorithm.

\subsection{PM Cluster Formation (PM-CF) stage}
As the name suggests, in this stage a set of PMs is chosen, known
as PM cluster. The global view of the entire data center network
must be obtained before processing the incoming VNs, essentially
in real-time manner. The global information includes the available
server resources such as remaining number of CPUs, memory and
storage availability. This also includes the bandwidth information
in each link. On the other hand, the information regarding the VN
includes the computing resource requirement such as memory, CPU
requirement. For exchange of intermediate data among VMs, the
required bandwidth must be made available. The required bandwidth
for each VM is calculated as the sum of bandwidth requirement of
each virtual link that is attached to the corresponding VM as
defined in Eq. (\ref{eq:def:vmBWreq}). As discussed earlier in Eq.
(\ref{eq:def:pmReqBW}), the remaining available bandwidth at each
PM is calculated by finding the minimum remaining bandwidth
available at each path that is originated from the corresponding
PM to all core switches.

Algorithm \ref{algo:algo1_clusterFormation} receives one VNE request at a time and the maximum bandwidth requirement among all VMs is calculated in Line \ref{algo:cluster:getMaxMem} - \ref{algo:cluster:getMaxNwReq} considering the resource requirement of each VM. The
remaining available bandwidth at each PM is then compared with the
maximum bandwidth requirement in order to remove PMs with
inadequate network resource as given in Line
\ref{algo:cluster:form_nw}. In this process, some PMs may fail to
fulfill the maximum bandwidth requirement among all the VMs.
However, such PMs may be attached to some physical links with
enough bandwidth resources to fulfill the demand of some VMs that
require minimum bandwidth among all VMs. We calculate the maximum
bandwidth requirement among all VMs as our goal is to select a set
of PMs that can fulfill the network resource demand of any VM.

From the resultant list of the PMs, the available server resource
and the resource requirement of the VMs are compared. The maximum
memory requirement and the maximum number of CPU requirements are
calculated. Based on this calculation, multiple PMs are removed
further from the list. Now the resultant list of PMs is eligible
to host any VMs from the current VN. Preparing the list of such PMs would reduce the embedding time and also ensures that any VM can be embedded onto any PM. The list is further sorted in
an ascending order based on any available resource type. Since,
all PMs are eligible to host all the VMs, $c$ number of PMs are
selected to form the PM cluster. The value of $c$, the PM cluster
size, is calculated based on the number of VMs that are present in
the current VN.

The details of the PM-CF stage are mentioned in Algorithm
\ref{algo:algo1_clusterFormation}. The PM-CF algorithm starts by
initializing $P$ and $V$ as set of all PMs and the inter-connected VMs
that belong to the current VN, respectively, as in Line
\ref{algo:cluster:init_P} - \ref{algo:cluster:init_V}. Following
this initialization, in Line
\ref{algo:cluster:getMaxMem}-\ref{algo:cluster:getMaxNwReq},
maximum resource requirement among all VMs is calculated. In the
proposed algorithm, resource requirement includes memory, CPU, and
network bandwidth requirement. The set $\ddot{P}_1$ is used to
keep a list of PMs that are eligible to provide enough network
bandwidth to all the VMs as in Line \ref{algo:cluster:form_nw}.
Similarly, the set $ \ddot{P}_2 \text{ and }\ddot{P}_3$ are used
to keep a list of PMs that are eligible to provide required amount
of computing resources as in Line \ref{algo:cluster:form_M} and
\ref{algo:cluster:form_C}. In Line \ref{algo:cluster:sort_nw_M_C},
the sorting operation is applied to the intersection of set
$\ddot{P}_1, \ddot{P}_2$ and $\ddot{P}_3$. From the resultant
sorted list of PMs $R$, first $n*c$ numbers of PMs are chosen, as
in Line \ref{algo:cluster:choose_final_cluster}. Here, $n$
represents the number of VMs present in the current VN, which is
calculated in Line \ref{algo:cluster:calcN}. We assume that the
value of $c$ is  decided by CSP satisfying the following
conditions.

\begin{equation}
1 \le c \le \frac{|\ddot{P}_1 \cap \ddot{P}_2 \cap \ddot{P}_3|}{n}
\end{equation}

Larger value of $c$ increases the PM cluster size. As a result,
larger number of PMs will be involved concurrently in VM pick-out
stage.

\begin{algorithm}\label{algo:algo1_clusterFormation}
    \linespread{1}\selectfont
    \SetAlgoLined
    \textbf{Initialize: } P = \{Set of all PMs\} \label{algo:cluster:init_P}\;
    \qquad \qquad  V = \{Set of all VMs\}  \label{algo:cluster:init_V}\;

    $n = |V|$\; \label{algo:cluster:calcN}
    $VN_{max}^{mem} = \max (\alpha_i^{mem}), \forall i, 1\le i \le n$\; \label{algo:cluster:getMaxMem}
    $VN_{max}^{CPU} =  \max (\alpha_i^{CPU}), \forall i, 1\le i \le n$\; \label{algo:cluster:getMaxCPU} 
    $\alpha_i^{nw} = \sum_{j=1}^{n} \alpha_{ij}^e, i\ne j$ \; \label{algo:cluster:getBWreq} 
    $VN_{max}^{nw} =  \max (\alpha_i^{nw}), \forall i, 1\le i \le n$\;\label{algo:cluster:getMaxNwReq} 

    $\ddot{P}_1 = \{p_i | \beta_i^n \ge VN_{max}^{nw}, \forall p_i \in P\}$ \label{algo:cluster:form_nw}\;
    $\ddot{P}_2 = \{p_i | \beta_i^{mem} \ge VN_{max}^{mem}, \forall p_i \in P\}$  \label{algo:cluster:form_M}\;
    $\ddot{P}_3 = \{p_i | \beta_i^{CPU} \ge VN_{max}^{CPU}, \forall p_i \in P\}$    \label{algo:cluster:form_C}\;

    $R =$ Sort $\{\ddot{P}_1 \cap \ddot{P}_2 \cap \ddot{P}_3 \}$ in ascending order \label{algo:cluster:sort_nw_M_C} \;
    $CL = $\{Set of first $n*c$ number of PMs from R\} \label{algo:cluster:choose_final_cluster} \;
    \caption{PM Cluster Formation (PM-CF) algorithm}

\end{algorithm}

\subsection{VM Pick-out (VM-Po) stage}

VM pick-out algorithm works in a distributed fashion. The central
server forwards separate copies of the details of the resultant PM
cluster and the VN to each PM. Each PM executes the VM-Po
algorithm with own copy of inputs. The details of resultant
PM cluster include the load of each PM, the number of networking
devices present between each pair of PMs, available amount of
computing resources (memory and CPU) on each PM, failure
probability of each PM etc. In order to reduce the network failure
impact onto current VN, multiple VMs can be placed onto a single PM.
As the physical network link at the upper layer of the Fat-Tree
topology carries out most of the data transmission, failure of
such physical links brings maximum impact onto the entire network.
Considering this fact, our objectives are to minimize the required
number of switches, physical links and PMs for embedding the VN.
When the PMs are allowed to host multiple VMs of single VN, it is obvious that the total number of required PMs will be reduced. We have provided the proof to minimize the required number of physical links and physical switches in Lemma \ref{lemma:minSw_minLink}-\ref{lemma:minLinksReq}.

It is assumed that the failure probabilities of the PMs are different.
As mentioned in the second part of the objective function in
Eq. \ref{eq:objective}, minimizing the required number
of switches infers the minimization of the required number of
physical links, which is proved in Lemma
\ref{lemma:minSw_minLink}. As discussed earlier, the basic
difference between three types of neighbor PMs is the number of
switches involved in the communication path. The number of
switches involved in the path between edge layer neighbors is 1 as
in Eq. (\ref{eq:edgeLayerNeighborPM}), whereas the number of
switches involved in the path between aggregation layer neighbors
is $3$ and core layer neighbors is $5$ as in Eq.
(\ref{eq:aggrLayerNeighborPM}) and Eq.
(\ref{eq:coreLayerNeighborPM}), respectively.

\begin{algorithm}[ht]\label{algo:VM-Po}
    \linespread{1}\selectfont
    \SetAlgoLined
    \KwIn{  Resource requirement of VN,  \\
        \quad \quad Resource availability of all PMs, \\
            \quad \quad  Workload of all PMs,  \\
            \quad \quad  Failure probability of all PMs
    }
\textbf{Initialization}\;
$V = \{$Set of VMs belong to current
VN$\}$\;
    $OD$: unique id of the PM, where the algorithm is running \label{algo:VM-Po:getOwnID}\;
    $p^{od} = $ represent the current PM, where the algorithm is running \label{algo:VM-Po:ownMachine}\;
    Form the set ${PE}^e(p^{od})$ according to Eq. (\ref{eq:edgeLayerNeighborPM}).\; \label{algo:VM-Po:edgeLayerNeighborPM}
    $PE = \{p^{od} \cup {PE}^e(p^{od})\}$\;\label{algo:VM-Po:formPE}
    \tcc{Find the leader PM based on min failure probability.}  
    $pl = min_{\forall p_i \in PE} fp(p_i)$\; \label{algo:VM-Po:findLeader}
    \eIf{$pl=p^{od}$} { \label{algo:VM-Po:checkLeader}
        Goto Line \ref{algo:VM-Po:calcFP}
    }{
        Wait for further instruction from PM $pl$\;\label{algo:VM-Po:waitForLeader}
    }
    \tcc{Calculate total failure probability of leader PM}
    $FP(pl) = \sum_{\forall p_i \in PE}{fp(p_i)}$\; \label{algo:VM-Po:calcFP}
    Compare with other leader PMs \;
    $SP = $ leader PM with minimum $FP$ value as calculated in Step \ref{algo:VM-Po:calcFP}\; \label{algo:VM-Po:compareLeaders}
    \eIf{$SP = p^{od}$}{\label{algo:VM-Po:checkSP}
        Goto Line \ref{algo:VM-Po:callPickOutFunc}\;
    }{
        Wait for further instruction from PM $SP$\;
    }
    $status=$Pick-out($V$)\; \label{algo:VM-Po:callPickOutFunc}
    \If{$status \ne $allEmbedded}{
        \ForEach{$pl \in {PE}^a(p^{od})$}{
            Remove VMs that are already allocated\;
            Send instruction to PM $pl$ to execute Pick-out($V$) function\;
            \If{$V = $empty}{
                Terminate\;
            }
        }
        \ForEach{$pl \in {PE}^c(p^{od})$}{
            Remove VMs that are already allocated\;
            Send instruction to PM $pl$ to execute Pick-out($V$) function\;
            \If{$V = $empty}{
                Terminate\;
            }
        }
    }
    \caption{VM Pick-out (VM-Po) Algorithm}
\end{algorithm}

The Algorithm \ref{algo:VM-Po} starts by obtaining the unique
machine ID, $OD$, of the PM where the algorithm is running, as in Line~\ref{algo:VM-Po:getOwnID}. As the
algorithm in each PM is executed independently, the value of $OD$
in each PM is different. $p^{od}$ represents the PM, where the
algorithm is running as given in Line~\ref{algo:VM-Po:ownMachine}.
Each PM forms a set ${PE}^e$ of edge layer neighbor PMs as in
Line~\ref{algo:VM-Po:edgeLayerNeighborPM}. In each PM, the
elements of the set ${PE}^e$ are different. Since, the set of edge
layer neighbor PMs does not include the PM itself, another set
$PE$ is formed with help of union operation of the set ${PE}^e$
and the PM itself as in Line \ref{algo:VM-Po:formPE}. The reason
behind forming the set of $PE$ is to allow the VMs to be embedded
onto the nearby PMs. Further, one leader PM is chosen based on the
failure probability of all the PMs that belong to the set $PE$. In
other word, the failure of all PMs that are connected to the same
edge switch are compared to find the leader PM. In
Line~\ref{algo:VM-Po:findLeader}, the PM with minimum failure
probability value $fp$ is chosen for the leader. As all the PMs
have the information regarding the failure probability of each
PMs, no additional communication is needed to choose the leader.
Further, it is obvious that the resultant leader in each PM is
same. The leader PM calculates the accumulative failure
probability $FP$ by summing the individual neighbor PM's failure
probability $fp$ as mentioned in Line \ref{algo:VM-Po:calcFP}.
This allows us to consider the failure probability of other PMs
with the failure probability of leader PM during the embedding
process. On the other hand, the PMs that do not satisfy to be the
leader have to wait for further instruction from their
corresponding leader as in Line \ref{algo:VM-Po:waitForLeader}.

It is possible that there may be multiple leader PMs. Such
situation may occur, if the PM cluster contains multiple edge
switches. For each edge switch, there will be single leader PM.
Further, each leader PM obtains the set of other leaders, $pl$,
that are connected to different edge switches. After obtaining the
set of leaders and their corresponding accumulated failure
probability value $FP$, the leader PM $SP$ with minimum
accumulative failure probability is chosen as mentioned in Line
\ref{algo:VM-Po:compareLeaders}. The leader PM $SP$ invokes the
function $Pick-out$ as in Algorithm \ref{algo:VM-Po} with the
argument as the set VMs or the current VN as mentioned in Line -
\ref{algo:VM-Po:callPickOutFunc}. The responsibility of the
$Pick-out$ function enables the corresponding PM to pick-out
suitable VM(s) without any conflict with other PMs. If the leader
PM $SP$ and the PMs that belong to the corresponding set $PE$
pick-out all the VMs present in the current VN, the embedding
process terminates. However, if any VM(s) is not picked-out by any
PM, other leader PMs that are connected to the same pod or the PMs
that are aggregation layer neighbors of the leader PM $SP$ will
invoke the function $Pick-out$. The same procedure is repeated for
the leader PMs that are core layer neighbors to the leader PM
$SP$.

\begin{algorithm}\label{algo:po}
    \linespread{1}\selectfont
    \SetAlgoLined
    \If{$p^{od}=pl$}{
        Send instruction to all PM that belong to set PE to execute Pick-out($V$) function\;
    }
    $p_c = min_{\forall p_i \in PE} {\Psi_i}$\; \label{algo:po:findPM}\tcc{Choose PM with minimum load from the set PE}
    Choose VM pair $(v_1,v_2)$ with maximum bandwidth demand \;\label{algo:po:findVMPair}
    $y_a = $ VM with maximum resource requirement between VM $v_1$ and $v_2$\; \label{algo:po:chooseFirstVM}
    \eIf{$\beta_c^x \ge \alpha_a^x$}{
        Assign $y_a$ to PM $p_c$\; \label{algo:po:assignFirstVM}
    }{
        Goto Line \ref{algo:po:findPM}\;
    }
    $y_b = $ VM with minimum resource requirement between VM $v_1$ and $v_2$\; \label{algo:po:chooseSecondVM}
    \eIf{$\beta_c^x \ge \alpha_b^x$}{
        Assign $y_b$ to PM $p_c$\;\label{algo:po:assignSecondVM}
    }{
        Goto Line \ref{algo:po:findPM}\;
    }
    \eIf{$V$ set is Empty}{
        Terminate the $Pick-out$ algorithm.
    }{
        Repeat from Line \ref{algo:po:findVMPair}\;
    }

    \caption{$Pick-out$ algorithm}
\end{algorithm}

The details of $Pick-out$ function are mentioned in Algorithm
\ref{algo:po}. Similar to VM-Po algorithm, each PM executes the
$Pick-out$ algorithm in a parallel manner. The algorithm starts by
receiving the VN as the only argument. Instead of all PMs that are
selected by the PM-CF algorithm, a set of edge layer neighbor PMs
executes concurrently at any given point of time. The algorithm is
named as $Pick-out$ as each PM chooses one or more than one
suitable VM independently from a set of VMs without any
communication with other PMs. The beauty of this algorithm is that
no VM will be chosen by multiple PMs without any communication
among them. The PM with minimum workload as given in Line
\ref{algo:po:findPM} picks out the VMs pair that requires maximum
network bandwidth as in Line \ref{algo:po:findVMPair}. From the
chosen VM pair, the VM with maximum resource requirement as in
Line \ref{algo:po:chooseFirstVM} is chosen by the PM as mentioned
in Line \ref{algo:po:assignFirstVM}. Similarly, the other VMs are
picked-out by the PM based on its remaining resources as mentioned
in Line
\ref{algo:po:chooseSecondVM}-\ref{algo:po:assignSecondVM}. The
process of picking VMs is repeated by the PM from Line -
\ref{algo:po:findVMPair} till the available resource is enough to
host a VM. All  PMs follow the same procedure in parallel manner.

In the proposed VNE algorithm and objective function
derived in Eq. \ref{eq:objective}, we exploit the
opportunity to embed multiple VMs onto a single PM. Assigning the
VMs onto an equal number of PMs require higher number of physical
links and thus higher number of switches taking the failure
probability of the PMs into account. In this scenario, failure of
a single PM brings maximum impact onto the single user as multiple
VMs of a VN are placed on the single PM. On the contrary, it is
observed that allocating multiple VMs onto a single PM could
entirely remove the requirement of the network bandwidth and
therefore could eliminate the possibilities of any impact onto the
corresponding virtual links due to the failure of the physical
links or switches. Failure of a PM would definitely have the
impact on either one or more users, which cannot be ignored.
However, the failure impact of the physical links and switches can
be ignored as multiple interconnected VMs can be allocated to one
PM.

The proposed algorithm is designed by considering the
characteristics of the Fat-Tree network topology, where multiple
layers of the core, aggregation, and edge layer switches are used
to connect a large number of physical servers. Based on this, the
numbers of neighboring PMs are calculated, which play an important
role in the embedding process. The approach of the proposed
algorithm can be applied to the data center with other network
topology. However, the proposed algorithm needs to be modified in
the environment with different network topology. For example, in
case of the server-centric network topology, the calculation of
the physical link and switch failure impact along with the number
of neighboring PMs need to be revisited.

As discussed before, the current version of FSC-VNE algorithm cannot be applied to other data center network topologies. Given that the FSC-VNE is modified and implemented in other data center network topologies, the algorithm may not suffer from the scalability issue. In the second stage, while picking out the VMs in a distributed manner, a large number of PMs can be employed to perform the pick-out task. However, a few numbers of PMs will be able to host the VMs based on their failure probability. Furthermore, scaling the proposed algorithm to a larger number of PMs, the embedding time may increase due to the communication overhead during the execution of the algorithm irrespective of the underlined network topology \ref{algo:VM-Po} and \ref{algo:po}.

\begin{lemma} \label{lemma:minSw_minLink}
Minimizing the number of switches eventually minimizes the
number of physical links required to embed a VN.
\end{lemma}

\begin{proof}
In 3-ary Fat-Tree network topology, the PMs are equipped with one
network port, i.e. the degree of the server is one. This indicates
that one PM can be connected to one network switch. Further, the
number of switches required to connect two PMs can be $1$, $3$, or
$5$. When two PMs share same edge switch, the length of the path
is $1$. When two PMs are attached to same pod but different edge
switches, the length of the path is $3$. Similarly, the length of
the path between two PMs attached to different pods is $5$.
Reducing the number of switches between two PMs refers to as ignoring
the use of either core switches or aggregation switches. If the
length of path between two PMs containing $5$ switches is reduced
to $3$, the number of links that connect to the core switch is
removed. As a result, the total number of required links is
reduced.
\end{proof}

\begin{lemma}\label{lemma:minSWReq}
Giving high preference to edge layer neighbors eventually reduces
the number of required physical switches.
\end{lemma}
\begin{proof}
As discussed earlier, edge layer neighbor PMs share common edge
switches. In Fat-Tree data center network topology, the degree of
server is $1$, which indicates the number of ports required on
each server. The PMs that share common edge switch have exactly
only one path of minimum distance. In Fat-Tree network topology,
the path  between any two PMs containing exactly two physical
links must have one switch. Hence, it is obvious that embedding
the VN onto the edge layer neighbor PMs reduces the number of
switches.
\end{proof}

\begin{lemma}\label{lemma:minLinksReq}
Giving higher preference to the edge layer neighbors eventually
reduces the number of required physical links.
\end{lemma}
\begin{proof}
Similar to the Lemma \ref{lemma:minSWReq}, the path between two
PMs that are connected to same edge switch contain exactly two
physical links, which is the minimum number of physical links
between any two PMs. Hence, embedding VMs onto multiple PMs that
are connected through one edge switch requires minimum number of
physical links.
\end{proof}

\begin{theorem}
Reducing the number of PMs considering their failure probability
reduces the VN failure probability based on the proposed VNE
algorithm.
\end{theorem}

\begin{proof}
Let, the virtual network $VN$ be embedded onto the set of PMs $p'
= \{p_1, p_2, \dots p_a\}, 1\le a \le |CL|$. The failure
probability of $VN$ can be written as
\begin{equation}
    fp'(VN) = \sum_{\forall p_i \in p'}\frac{1}{|p'|}{fp(p_i)}
\end{equation}
    Based on the Algorithm \ref{algo:VM-Po}, the PMs with lesser failure probability are given higher preference to pick-out the VMs. Hence, the following relation can be derived.
    \begin{equation}
    fp(p_1) < fp(p_2) < \dots < fp(p_a)
    \end{equation}
    If embed the same virtual network $VN$ can be embedded onto lesser number of PMs, the resultant PM set would be $p'' = \{p_1, p_2, \dots p_b\}, b < a$.
    After embedding the $VN$ onto the PM set $p''$, the new failure probability of $VN$ would be $ fp''(VN) = \sum_{\forall p_j \in p''}\frac{1}{|p"|}{fp(p_j)}$. The above failure probability inequality can be re written as
    \begin{equation}
     fp(p_1) < fp(p_2) < \dots < fp(p_b)  < \dots < fp(p_a)
    \end{equation}
    Which implies that,
    \begin{align}
     \nonumber \sum_{\forall p_i \in p''}\frac{1}{|p"|}{fp(p_i)}  < \sum_{\forall p_j \in p'}\frac{1}{|p'|}{fp(p_j)} \\
    \Rightarrow fp''(VN) < fp'(VN)
     \end{align}
\end{proof}

\begin{theorem}
The total impact of physical links failure on virtual links can be
minimized by reducing the number of PMs.
\end{theorem}
\begin{proof}
Let, the virtual network $VN$ consists of $a$ number of VMs and
$b$ number of virtual links. According to the proposed FSC-VNE
algorithm, the PMs with enough available resources pick-out
multiple VMs. Let, VM $v_q, 1 \le q \le a$ and $v_r, 1 \le r \le
a, q\ne r$ be picked-out by a single PM. In such scenario, the
bandwidth demand of the virtual link between the VM $v_q$ and
$v_r$ can be ignored. Hence, the value of $\lambda_{ij}^{qr}$ can
be rewritten as given in Eq. \ref{eq:lambdaDef}.
    \begin{equation}
        \lambda_{ij}^{qr}=0, \text{ when } \kappa_i^q =1, \kappa_j^r =1, i=j
    \end{equation}
Further, applying above equation onto the Eq.
\ref{eq:NWfailImpactOnVirtualLink} for all virtual links, the
maximum impact onto the virtual links can be minimized to less
than 1.
\end{proof}

\begin{lemma} \label{lemma:CommOverhead}
The communication overhead of the proposed FSC-VNE algorithm is
$(n\partial - 1)(m'-1)$, where $\partial$ be the amount of memory
required to store the current status of one VM and $m'=|CL|$.
\end{lemma}

\begin{proof}
The VNs are embedded in two stages according to the proposed
algorithm. As the first stage follows the centralized approach,
there is no communication overhead. However, in second stage, the
leader PM with minimum total failure probability sends the updated
information regarding all the VMs with other edge layer neighbor
PMs and other neighbor PMs. Let, $\partial$ be the amount of
memory required to store the current status of one VM. Current
status infers, if a VM is picked-out by any PM. For $m' = |CL|$
number of selected PMs, a total of $2(m' -1)$ number of
communications are needed for sending and receiving the messages.
The leader PM sends a total of $n*\partial*(m'-1)$ number of data
packets to all neighbor PMs. The leader PM also receives $m'-1$
number of single data packet messages. Hence, to embed a VN with
$n$ VMs, a total of $(n\partial - 1)(m'-1)$ number of data packets
are transmitted over the physical network.
\end{proof}

\begin{theorem}
The proposed semi-centralized approach makes the VNE process
faster as compared to that of the centralized one.
\end{theorem}

\begin{proof}
In both centralized and semi-centralized approach, each VN must
wait before they are submitted to the embedder.
However, the major difference in both aforementioned approaches is
the time that every VN needs to spend in the waiting queue. In the
centralized version, the VN needs to wait until all the VMs of the
previous VN are assigned to the suitable PMs unlike in the
semi-centralized approach. \rtwo{Let, $t^c$ units of time be taken
in the centralized algorithm to embed a virtual network of $n$
numbers of VMs. In each time slot, let $v^T$ number of VNs be
arrived to the system. Since, all the VNs embedding will be
carried out by the single server, the total waiting time of the
VNs that are arrived in the first time unit can be calculated as
$\frac{1}{2}[t^c*v^T*(v^T+1)]-v^T$. Choosing a set of suitable PMs
for the entire VN without considering each VM's resource
requirement would take lesser time, which is implemented in the PM-CF
stage. Let, $t^{s'}$ be the time required to select a set of PMs
for each VN. Hence, each VN needs to wait until the previous VN is
processed in the PM-CF stage. Mathematically, $t^c > t^{s'}$. }

\rtwo{This infers that a larger number of VNs are being processed
by the server at any particular instant of time in the PM-CF stage
as compared to the server in the centralized approach. Further,
multiple PMs are assigned to embed all the VMs in the VM-Po stage.
Let, $t^{s"}$ be the time taken to embed one VM. A VN of $n$
number of VMs would take a total of $n*t^{s"}$ units of time, if
all the VMs are processed in a sequential manner. Since, multiple
PMs are involved in embedding all the VMs in a parallel manner,
the total time required to embed $n$ number of VMs would take
$t^{s"}$ units of time. Combining the time required to embed one
VN consisting of $n$ number of VMs in the PM-CF and VM-Po stages,
the total time can be calculated as $(t^{s'} + t^{s"})$. It is
obvious that $(t^{s'} + t^{s"}) < t^c$ as $t^c$ involves the
embedding time of all VMs in a sequential manner and $(t^{s'} +
t^{s"})$ involves in embedding all the VMs in parallel manner.
Taking the above scenario into consideration, the total waiting
time of the VNs that are arrived in the first time unit can be
calculated as $t^{c'}*(v^T-1)$, where $v^T$ number of VNs arrives
into the system in a single time unit. In the above calculation,
the time required to embed in the VM-Po stage is not taken into
consideration as different sets of PMs are selected for each VN in
the PM-CF stage. This indicates that the VNs are embedded onto the
physical network in the proposed semi-centralized approach by
taking lesser time, which makes the embedding process faster.}
\end{proof}

\begin{theorem}\label{lemma:complexity_PM_CF}
The time and space complexity of PM-CF algorithm are
$\mathcal{O}((m\log{m} )+ n)$ and  $\mathcal{O}(m+ n^2)$,
respectively , where $m$ represents the number of PMs and $n$ is
the number of VMs in the VN.
\end{theorem}

\begin{proof}
In PM-CF algorithm, initialization instruction as mentioned in
Line \ref{algo:cluster:init_P}-\ref{algo:cluster:calcN} takes
constant time $\mathcal{O}(1)$. Calculating the maximum resource
requirement as in
Line~\ref{algo:cluster:getMaxMem}-\ref{algo:cluster:getMaxNwReq},
it takes $\mathcal{O}(3n)$ running time as we are only considering
three types of resource requirement. The time complexity for
forming the set $\ddot{P}_1, \ddot{P}_2, \text{and} \ddot{P}_3$ in
Line \ref{algo:cluster:form_nw}-\ref{algo:cluster:form_C} is
$\mathcal{O}(3n)$ time as in the worst case scenario, all PMs may
satisfy the resource constraint. The time required to sort the set
$\ddot{P}_1 \cap \ddot{P}_2 \cap \ddot{P}_3$ as in Line
\ref{algo:cluster:sort_nw_M_C} can be minimized to
$\mathcal{O}(m\log{}m)$ by using any existing efficient sorting
algorithm. Hence, the total time complexity can be calculated as
\begin{equation}
T(\text{PM-CF}) = \mathcal{O}(3n) + \mathcal{O}(3m) + \mathcal{O}(m\log{}m)
\end{equation}
The above time complexity of PM-CF algorithm in worst case
scenario can be concluded as $ \mathcal{O}(m\log{}m + n)$.

Let, $\delta$ be the constant amount of the space required to
store the details of each PM. The total space required to store
the entire PM network is $\mathcal{O}(\delta*m)$. Further
considering the link representation of the VN, which includes the
virtual node and virtual link representation, the space complexity
of the PM-CF algorithm is $\mathcal{O}(n^2)$. Hence, the space
complexity of the PM-CF algorithm can be concluded as
$\mathcal{O}(m+ n^2)$.
\end{proof}

\begin{theorem}\label{lemma:complexity_VM-Po}
The time and communication complexity of the VM-Po algorithm are
$\mathcal{O}(m'+n^3)$ and $\mathcal{O}(2(m'-1))$, respectively.
\end{theorem}

\begin{proof}
VM-Po algorithm executes in $m' = |CL|$ number of PMs in a
distributed approach. Before calculating the time complexity of
the VM-Po algorithm, it is necessary to analyze the running time
complexity of Algorithm~\ref{algo:po} as it is invoked from the
VM-Po algorithm for multiple times.

In Algorithm \ref{algo:po}, finding PM with minimum load as in
Line \ref{algo:po:findPM} has the time complexity of
$\mathcal{O}(m')$. Assuming that the link representation is used
to keep the information of VN, in the worst case,
$\mathcal{O}(n^2)$ numbers of comparisons are required to find the
VM pair with maximum bandwidth, which is repeated for
$\mathcal{O}(\frac{n}{2})$ number of times. Here, picking out the
VM from the selected virtual edge takes a constant time. Hence,
the total running time of the Pick-out algorithm can be concluded
to $\mathcal{O}(n^2 * \frac{n}{2}) \approx \mathcal{O}(n^3)$.

In Algorithm \ref{algo:VM-Po}, forming the set of edge layer PMs
in Line \ref{algo:VM-Po:formPE} requires the running time of
$\mathcal{O}(m')$. The running time for calculating the total
failure probability as in Line \ref{algo:VM-Po:calcFP} is
constant. Following this as in Line
\ref{algo:VM-Po:compareLeaders}, finding the leader with minimum
$FP$ value takes $\mathcal{O}(m')$ running time. The two
\textit{for} loops require $\mathcal{O}(m')$ computational time.
As the running time of Pick-out algorithm is $\mathcal{O}(n^3)$,
the total time complexity of the VM-Po algorithm can be concluded
to $\mathcal{O}(m' + n^3)$.

The communication complexity of the VM-Po algorithm can be seen as
$\mathcal{O}(2(m'-1))$ as discussed in Lemma
\ref{lemma:CommOverhead}. The leader PM that is calculated in Line
\ref{algo:VM-Po:calcFP} invokes the Pick-out algorithm and sends
the instruction to other $m'-1$ number of PMs in order to execute
the Pick-out algorithm. Upon execution of the Pick-out algorithm,
$m'-1$ number of messages are received by the leader PM. As a
result, the total number of communications established in second
stage is $2(m'-1)$.
\end{proof}

\section{Performance Evaluation} \label{sec:perfEvaluation}

In this section, we evaluate the performance of proposed FSC-VNE
algorithm by using a discrete event Java-based simulator. We
compare our algorithm against two popular centralized VNE
algorithms VNE-DCC \cite{2016-5} and VIE-SR \cite{2015-2}
considering the performance metrics such as acceptance ratio and
average number of required resources.

\subsection{Simulation Setup}
In the simulation environment, one data center is formed with a
number of PMs that are connected by the Fat-Tree network
topology. The simulation is performed on the small scale
data center, where each switch in the data center is equipped with
$k=32$ number of ports. As discussed in Section
\ref{sec:sysModel}, the maximum number of PMs that can be
connected with $32$ ports switches is $m=\frac{32^3}{4} = 8192$
and the total number of switches required is $1280$ including the
core, aggregation, and edge switches. In the simulation, it is assumed that the failure events are independent and random. These events also do not occur frequently. Based on such properties, the number of failure events is modeled as a Poisson process with average number of failures ($\lambda$) $0.01$. The value of lambda is chosen to be very less to reflect the practical scenarios, where the number of failures occurring in the actual data center is relatively low. In the simulation, assignment of the physical resources such as CPU, memory, and storage follows
the Random distribution. The maximum amount of storage and memory
of a PM ranges between $1000 GB - 2000 GB$ and $20GB - 50GB$,
respectively. The number of processing units in each PM is
randomly distributed between $8$ through $32$. The maximum
bandwidth capacity on each physical link is $100$ Gbps.

The number of VMs in each VN ranges between $2$ through
$10$.\rthree{The arrival rate of the incoming virtual network
follows the Poisson distribution with the mean arrival rate of
$10$ VNs per $100$ units of time. The mean arrival rate refers to
as the average number of VNs arriving in every $100$ time units.
Similarly, the lifetime of the VNs follows the exponential
distribution with mean lifetime of $250$ time units. The lifetime
of a VN refers to as the number of time units after which the VN
is terminated.} No specific topology is followed to construct a
VN. However, a probability value is assigned to each virtual link,
called virtual link probability. A VN with virtual link
probability of $1$ results in fully connected VN. The link
probability of VNs in the simulation environment ranges between
$0.5-0.8$. Similar to the physical network, virtual resources are
distributed among the VMs by following the random distribution.
The resource configurations of the VNs are not strictly
heterogeneous, which infers that some VNs can be homogeneous in
terms of their resource configurations. The CPU requirement for
each VM ranges from $1$ through $4$ numbers of CPUs. The required
amount of memory and storage ranges from $500 MB$ through $4 GB$
and $8GB$ through $20GB$, respectively. Similarly, the bandwidth
requirement for each VN ranges from $50 Mbps$ through $1 Gbps$.

The performance of the FSC-VNE algorithm is evaluated mainly in
three metrics: acceptance ratio, amount of resource allocated, and
embedding time. Acceptance ratio is defined as the ratio between
the number of VNs accepted and the number of VNs received by the
CSP from users. Amount of allocated resource  refers to as the number
of PMs and physical links required per VN. Embedding  time is the
time required to embed the VN onto physical network, which involves
the waiting time of VN and the time taken by the algorithm to
carry-out the embedding procedure. Taking the above-mentioned
performance matrices and the simulation environment, the following
simulation results are derived.

\begin{figure}[!hb]
    \begin{center}
        \epsfig{file=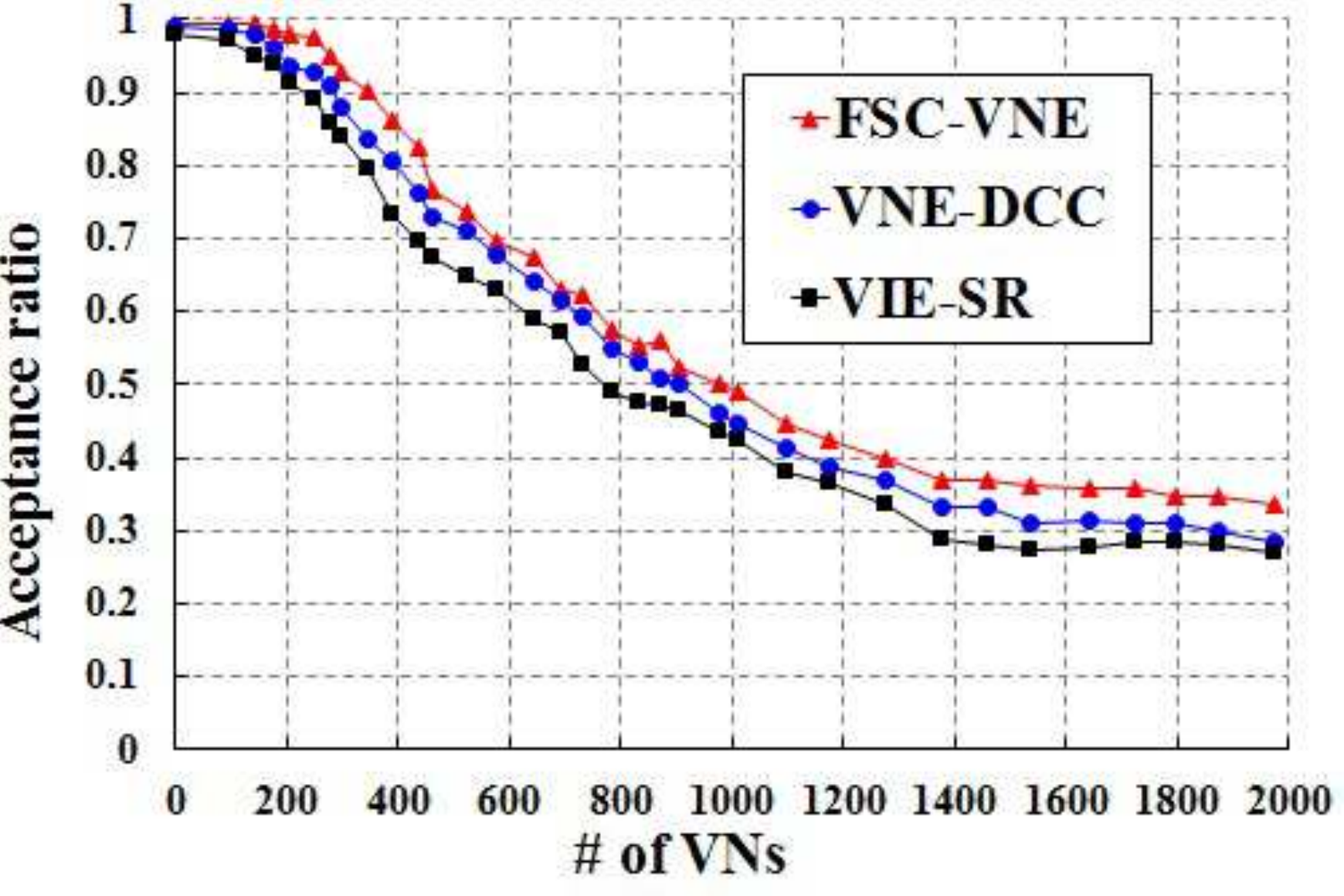,width=60mm}\vspace{-3mm}
        \caption{Acceptance ratio.}\vspace{-4mm}
        \label{fig:sim:acceptanceRatio}
    \end{center}
\end{figure}

\begin{figure*}[!h]
    \centering
    \subfigure[Average number of PMs required per VN]
    {
        \includegraphics[width=0.30\linewidth]{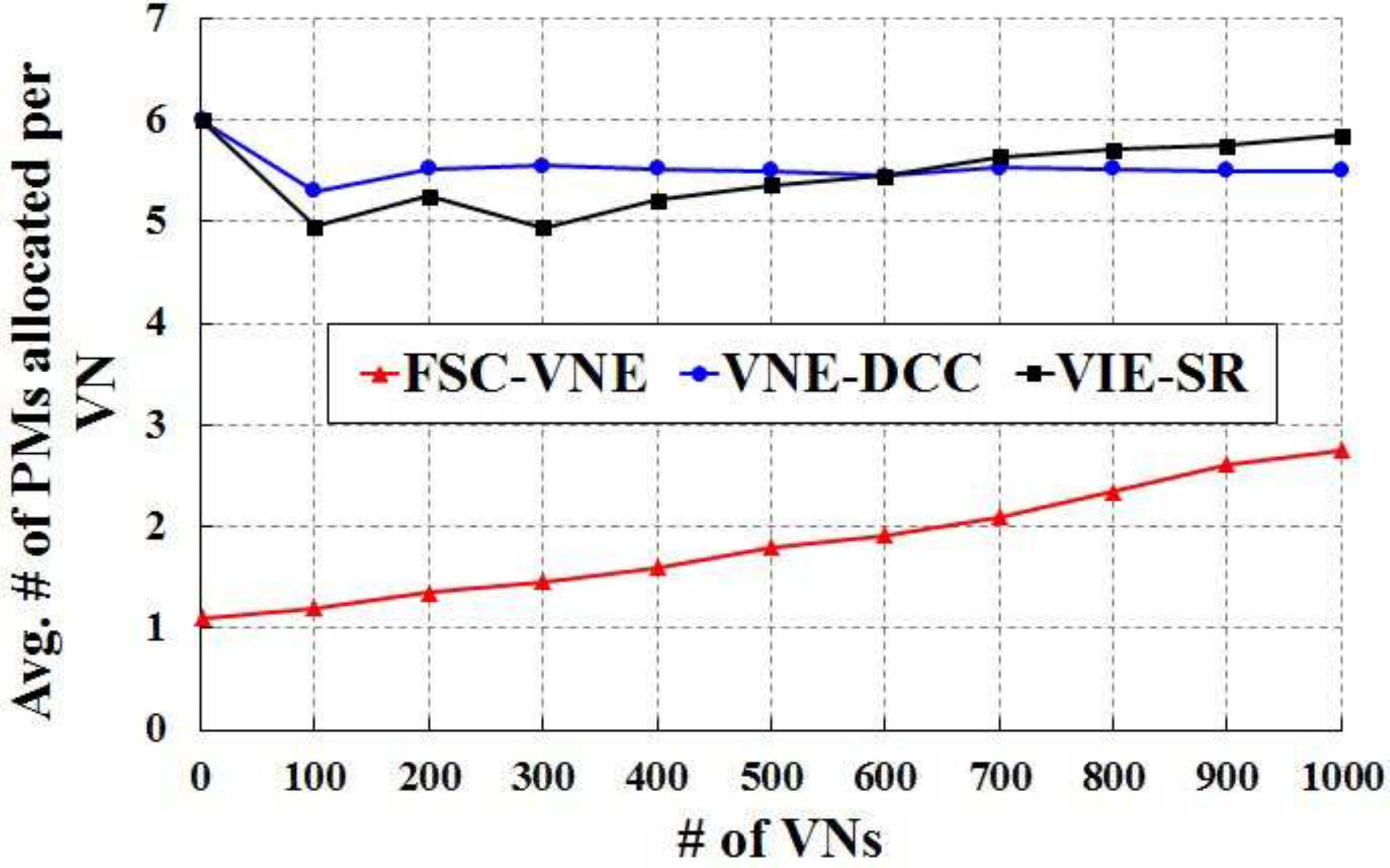}
        \label{fig:sim:avg_PMReq}
    }
    \subfigure[Average number of links required]
    {
        \includegraphics[width=0.30\linewidth]{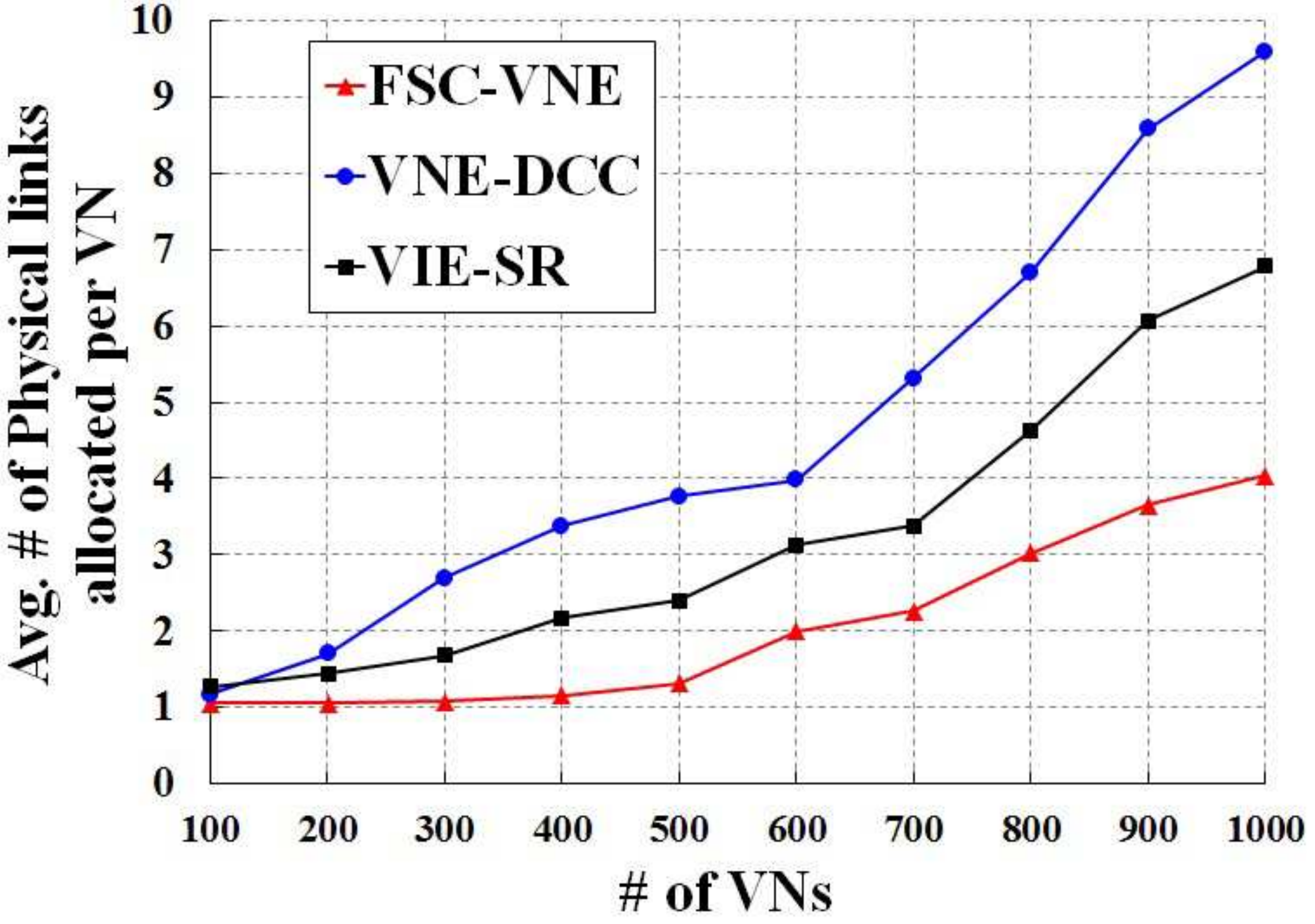}
        \label{fig:sim:avg_LinkReq}
    }
    \subfigure[Average number of switches required]
    {
        \includegraphics[width=0.30\linewidth]{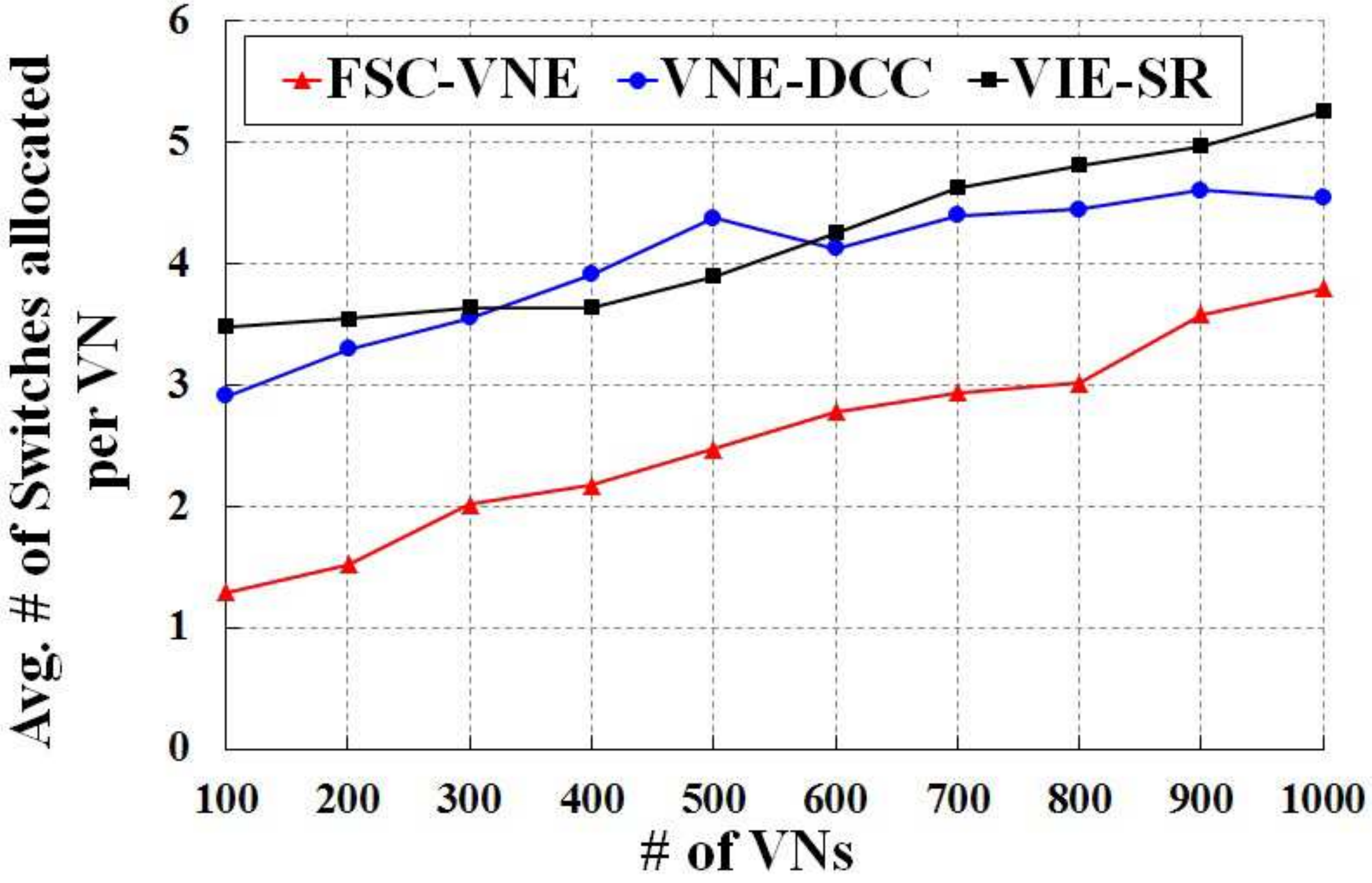}
        \label{fig:sim:avg_SwitchReq}
    }
    \vspace{-2mm}
    \caption{Average number of resource required.}\vspace{-5mm}
    \label{fig:sim:avg_NWresReq}
\end{figure*}
\subsection{Simulation Results}
The performance of proposed FSC-VNE algorithm is evaluated by
comparing with that of VNE-DCC and VIE-SR algorithm. Fig.
\ref{fig:sim:acceptanceRatio} demonstrates the relationship
between the number of VNs received and the acceptance ratio. It
can be observed that, with the increasing number of VNs, the
chances for a VN to be accepted decreases. In terms of acceptance
ratio, the proposed algorithm outperforms over other VNE
algorithms. In case of FSC-VNE algorithm, the acceptance ratio for
$2000$ VNs is about $0.33$, whereas in case of other algorithms
the acceptance ratio is less than $0.3$.

The relationship of number of VNs with the amount of allocated
physical resources is presented in Fig.
\ref{fig:sim:avg_NWresReq}. Here, the physical resource refers to
as the number of PMs, and allocated number of physical links and
switches. The proposed algorithm allows the multiple VMs from one
VN to be hosted by the single PM resulting in the reduction of the
number of required physical links and switches. Keeping the total
number of PMs constant, the average number of PMs allocated to
$100$ VNs is $1.1$, as shown in Fig. \ref{fig:sim:avg_PMReq}. The
number of PMs increases to $2.7$, when the number of VNs increases
to $1000$. On the contrary, it is observed that the average number
of PMs allocated is equal to the number of VMs in case of other
algorithms. The average number of PMs allocated to host $100$ VNs
is $5.5$ and $4.9$ in case of VNE-DCC and VIE-SR algorithm,
respectively. However, the average number of PMs required to embed
$1000$ VNs is $5.5$ and $5.8$ in case of VNE-DCC and VIE-SR
algorithm, respectively. Such variation along the Y-axis occurs
due to the random distribution of the number of VMs in each VN.

\begin{figure}[h]
    \centering
    \subfigure[Mean embedding time]
    {
        \includegraphics[width=0.48\linewidth]{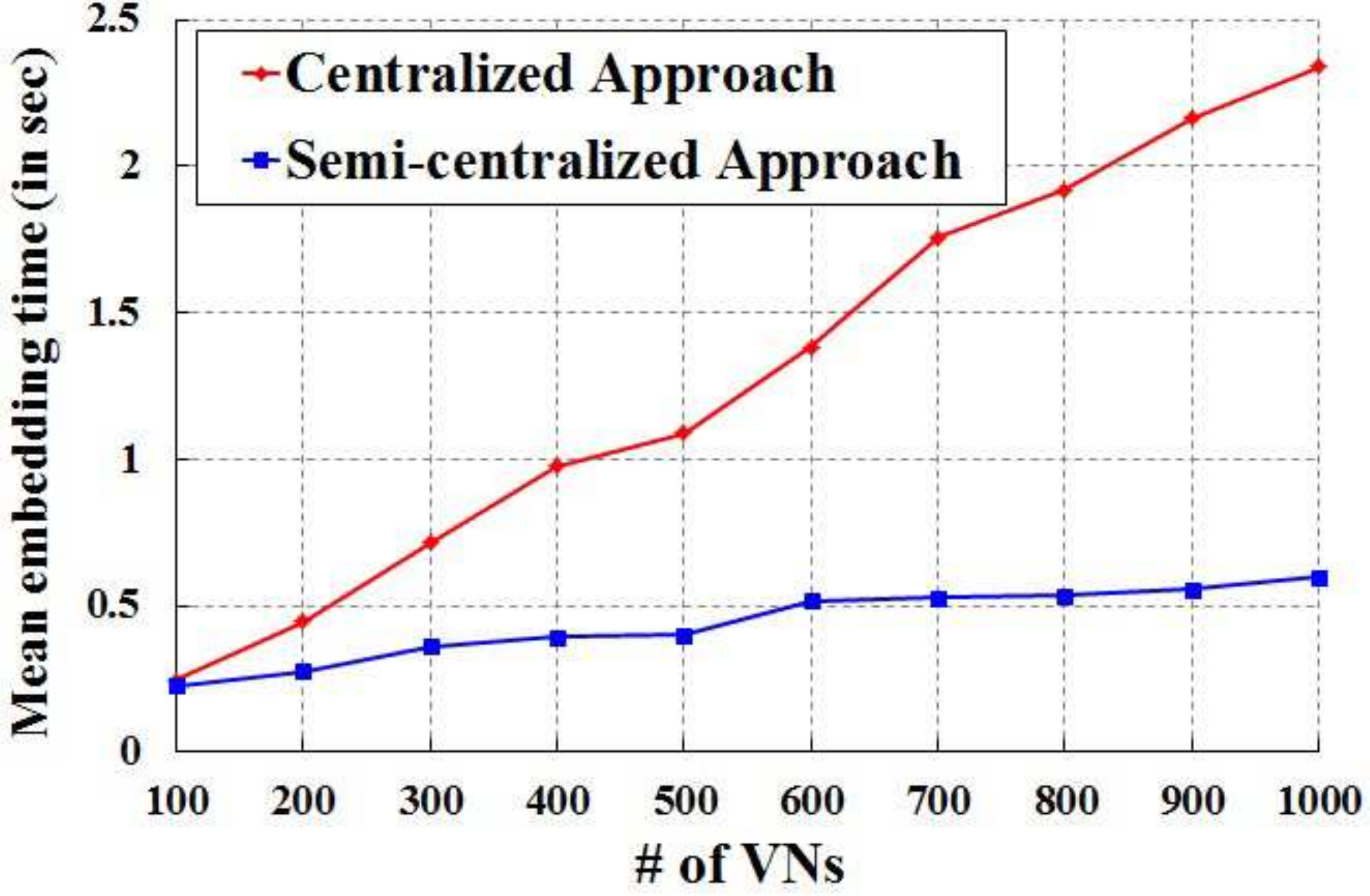}
        \label{fig:sim:embedTime_centra_semiCentra}
    }
    \subfigure[Mean waiting time]
    {
        \includegraphics[width=0.45\linewidth]{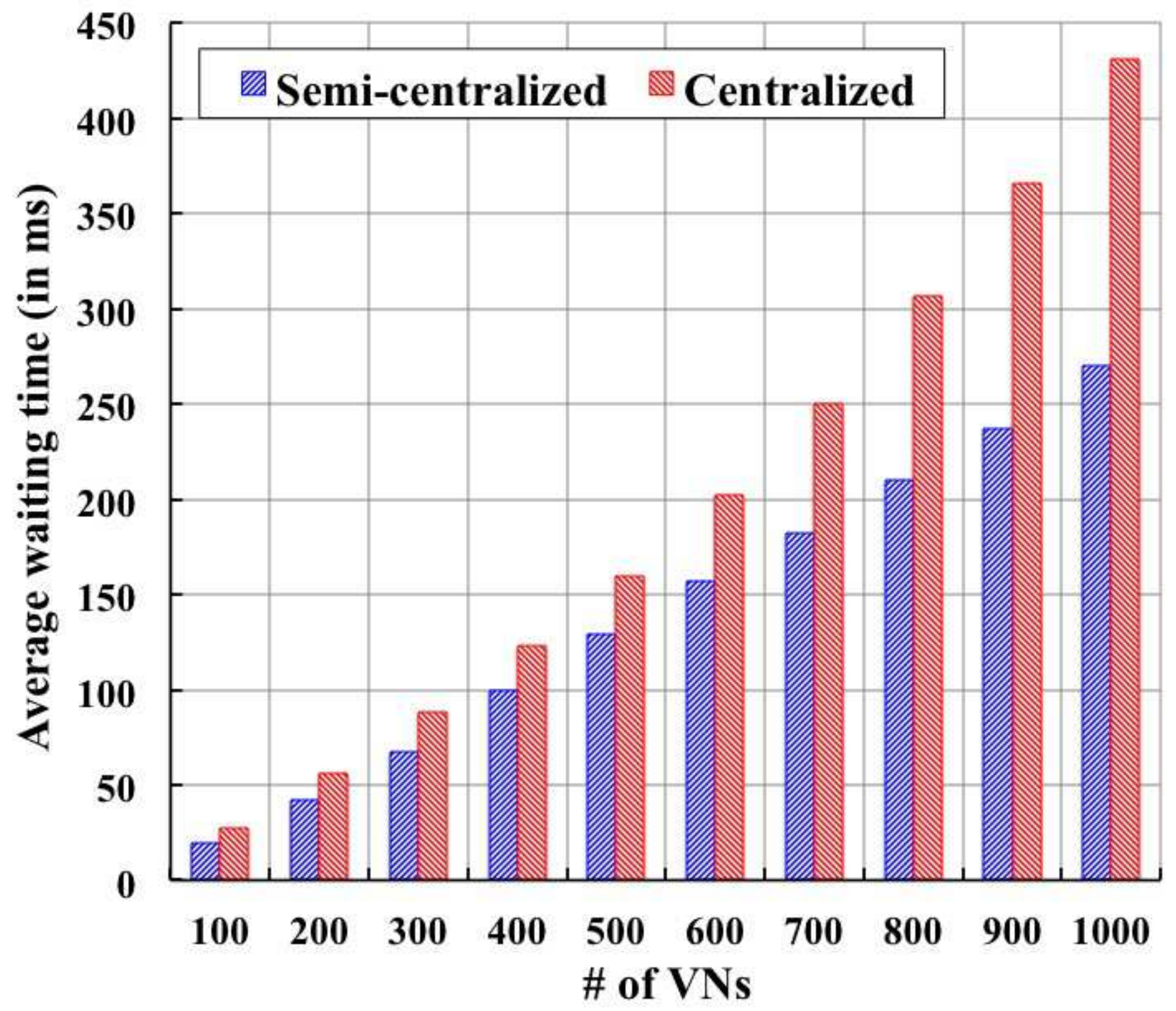}
        \label{fig:sim:cen_semiCent:waitTime}
    }
\\
    \subfigure[Mean throughput]
    {
        \includegraphics[width=0.47\linewidth]{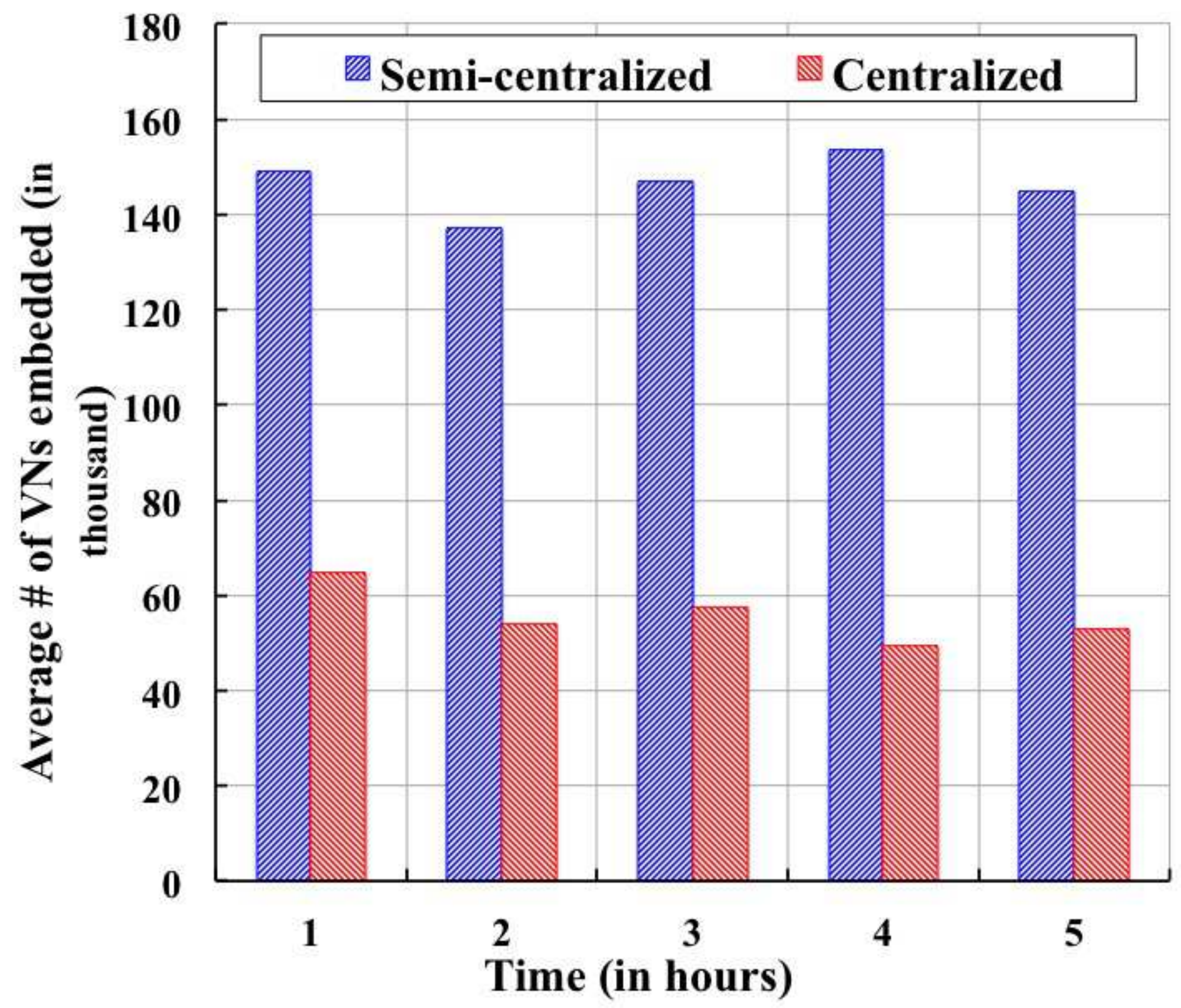}
        \label{fig:sim:cen_semiCent:throughput}
    }
    \subfigure[Mean failure probability]
    {
        \includegraphics[width=0.47\linewidth]{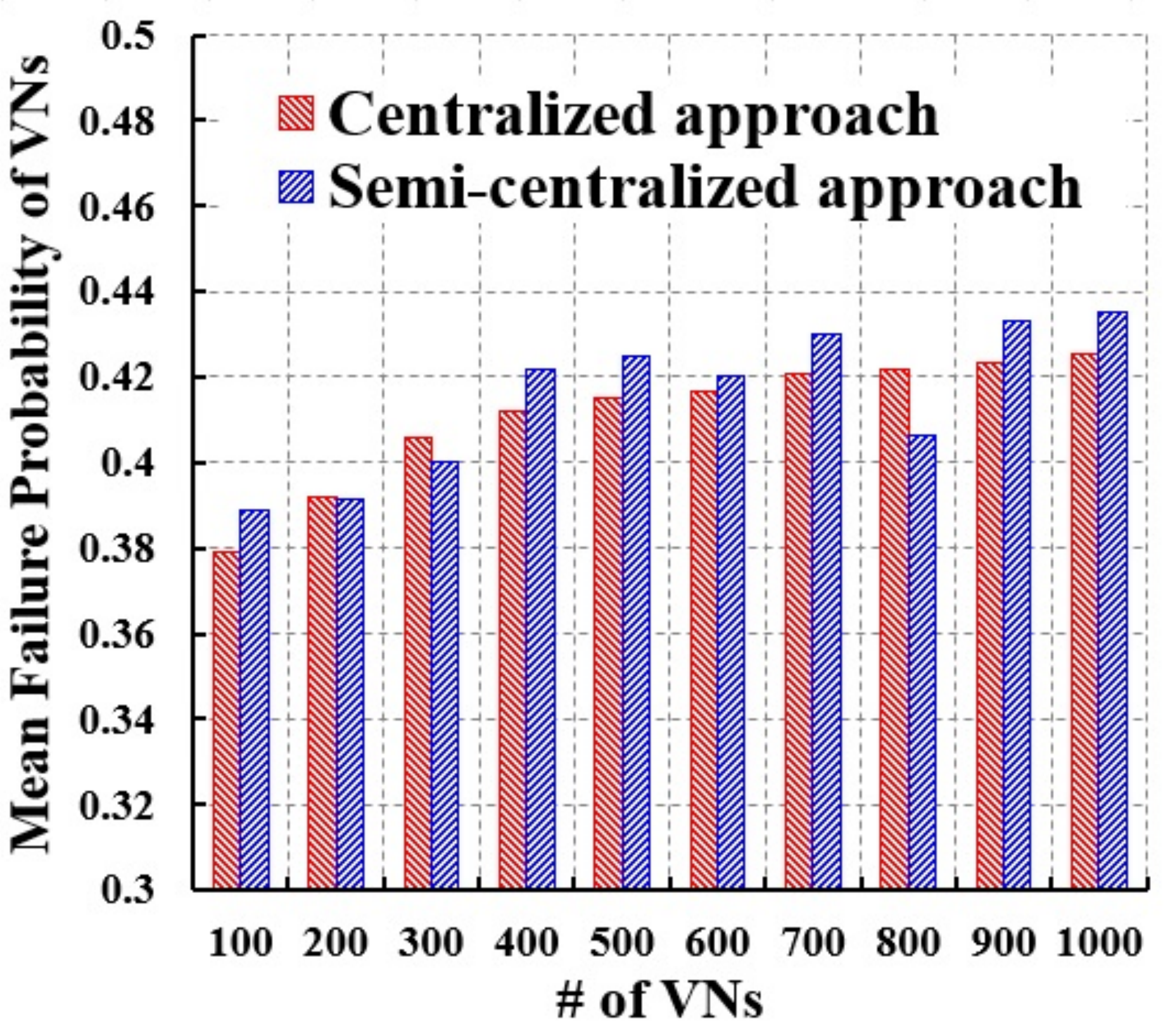}
        \label{fig:sim:cen_semiCent:meanfailprob}
    }
    \vspace{-2mm}
    \caption{Comparison between semi-centralized and centralized version of FSC-VNE.}
    \vspace{-4mm}
    \label{fig:sim:cen_semiCent}
\end{figure}

The advantage of embedding multiple VMs that belong to same VN
onto a single PM can be realized in Fig. \ref{fig:sim:avg_LinkReq}
and \ref{fig:sim:avg_SwitchReq}. For upto $100$ VNs the average
number of physical links required is approximately $1$  in case of
all three algorithms, as shown in Fig. \ref{fig:sim:avg_LinkReq}.
However, this number increases with the increasing number of VNs.
The performance of the proposed algorithm can be distinguished
from other algorithms easily when the number of VNs increases to
$1000$, where the required number of physical links is $4.0$,
$9.5$, and $6.78$ in case of FSC-VNE, VNE-DCC, and VIE-SR
algorithm, respectively. Similar trend is observed in Fig.
\ref{fig:sim:avg_SwitchReq}. The average number of switches
involved in embedding the VNs ranges through $100$ and $1000$ is
$1.2$ and $3.7$ switches, in case of the proposed FSC-VNE
algorithm. However, in case of VNE-DCC and VIE-SR algorithm, the
average number of switches required to embed $100$ VNs is $2.9$
and $3.4$, respectively. When the number of VNs increases to
$1000$, the average number of required switches  increases to
$4.5$ and $5.2$ switches, respectively.

Embedding time is one of the major performance metrics that needs
to be considered while analyzing the performance of the algorithm.
Here, embedding time refers to as the time taken by the algorithm
in order to embed the incoming VN. The proposed algorithm follows
the semi-centralized approach in order to fasten the embedding
process. Such approach allows the CSP to increase the number of
PMs without compromising the performance of embedding algorithm.
The comparison of performance of  centralized version and the
semi-centralized version of the proposed algorithm is presented in
Fig. \ref{fig:sim:embedTime_centra_semiCentra}. The number of PMs
is kept constant at $8192$ in the simulation. From the  simulation
results, it can be clearly observed that the embedding time in
case of semi-centralized approach is lesser than that of the
centralized approach. The average time required to embed $100$ VNs
is $0.24sec$, and $0.22sec$,  in case of centralized and
semi-centralized approach, respectively. However, the average time
required to embed $1000$ VNs is $2.34sec$ and $0.60sec$.

\begin{figure}[h]
    \begin{center}
        \epsfig{file=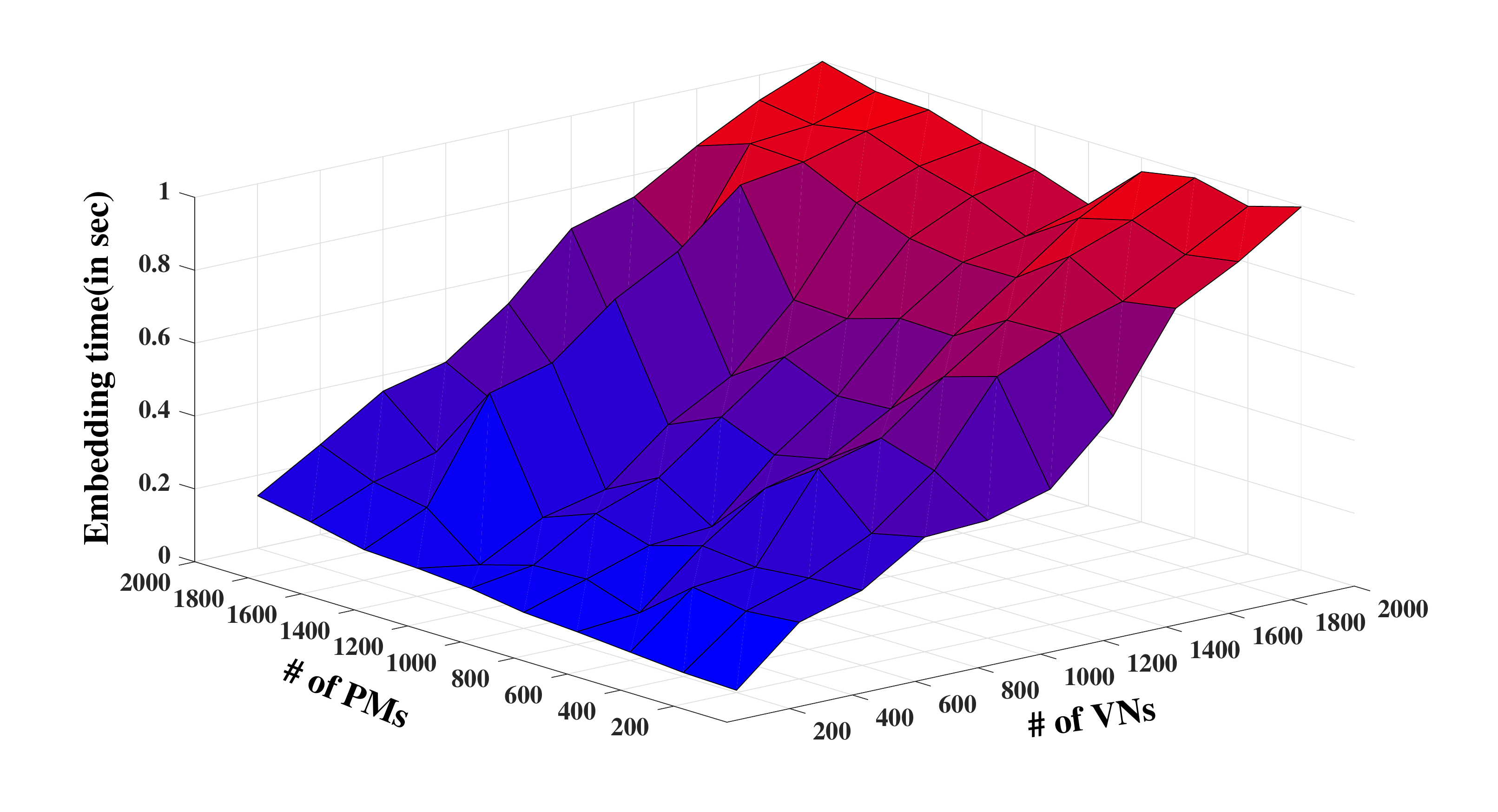,width=85mm}
       \vspace{-3mm}
        \caption{Embedding time with varying number of VNs and PMs.}
        \vspace{-5mm}
        \label{fig:sim:embedTime}
    \end{center}
\end{figure}

The small improvement in the simulation Fig.
\ref{fig:sim:embedTime_centra_semiCentra} is obtained from the
small scale simulated environment. However, the benefits of this
small improvement in embedding time can be realized in the
real-time environment, where the number of PMs can reach 1 million
and thousands of VNs may arrive at the data center in every
second.

Waiting time of VNs and throughput of the FSC-VNE are considered
as another matrix to evaluate the performance of the
semi-centralized approach of the VN embedding as shown in Fig.
\ref{fig:sim:cen_semiCent:waitTime} and
\ref{fig:sim:cen_semiCent:throughput}, respectively. In Fig.
\ref{fig:sim:cen_semiCent:waitTime}, waiting time refers to as the
time spent by the VN before embedding onto the physical network.
\rone{As shown in Fig.} \ref{fig:sim:cen_semiCent:waitTime},
\rone{average waiting time indicates if the proposed algorithm can
achieve its goal of minimizing the waiting time of the VNs and can
reduce the embedding time. Reducing the waiting time and fastening
the embedding process also improves the QoS of the cloud service
provider.} By embedding $100$ numbers of VNs, the average waiting
time in our proposed semi-centralized and centralized algorithms
is $20 ms$ and $27 ms$, respectively. In the semi-centralized
version, the mean waiting time increases to $270 ms$, when the
number of VNs increases to $1000$.

On the other hand, in the
centralized approach, the mean waiting time increases to more than
$400 ms$, when the number of VNs increases to $1000$. \rone{As
shown in Fig.} \ref{fig:sim:cen_semiCent:throughput},\rone{ the
throughput that refers to as the number of VNs embedded per unit
time, plays a major role in evaluating the performance of the
embedding algorithm as reduction in throughput directly affects
the users' experience.} The mean throughput of the FSC-VNE
algorithm is evaluated in each hour for 5 hours. It is observed
that the semi-centralized approach can embed more number of VNs as
compared to the corresponding centralized approach. In the entire
$5$ hours of time duration, the average number of embedded VNs
ranges between $140,000$ and $160,000$ using the semi-centralized
approach, whereas the average number of VNs ranges between
$50,000$ and $60,000$ using the centralized approach. From the
above-mentioned performance metrics as shown in Fig.
\ref{fig:sim:cen_semiCent}, it is observed that the
semi-centralized approach outperforms over the centralized one.
The centralized and semi-centralized approaches of the proposed
algorithm are also compared in terms of mean failure probability,
as shown in Fig. \ref{fig:sim:cen_semiCent:meanfailprob}. It is
observed that the proposed semi-centralized approach has
negligible impact onto the objective. In case of centralized
approach, the mean failure probability lies between $0.36$ and
$0.44$, when the number of VNs increases from 100 to 1000. Similar
trend is observed in case of semi-centralized algorithm. The mean
failure probability lies between $0.38$ and $0.44$ with the same
the number of VNs. It is due to the fact that semi-centralized
approach is adopted primarily to improve the performance of the
algorithm.

\begin{figure}[h]
    \begin{center}
        \epsfig{file=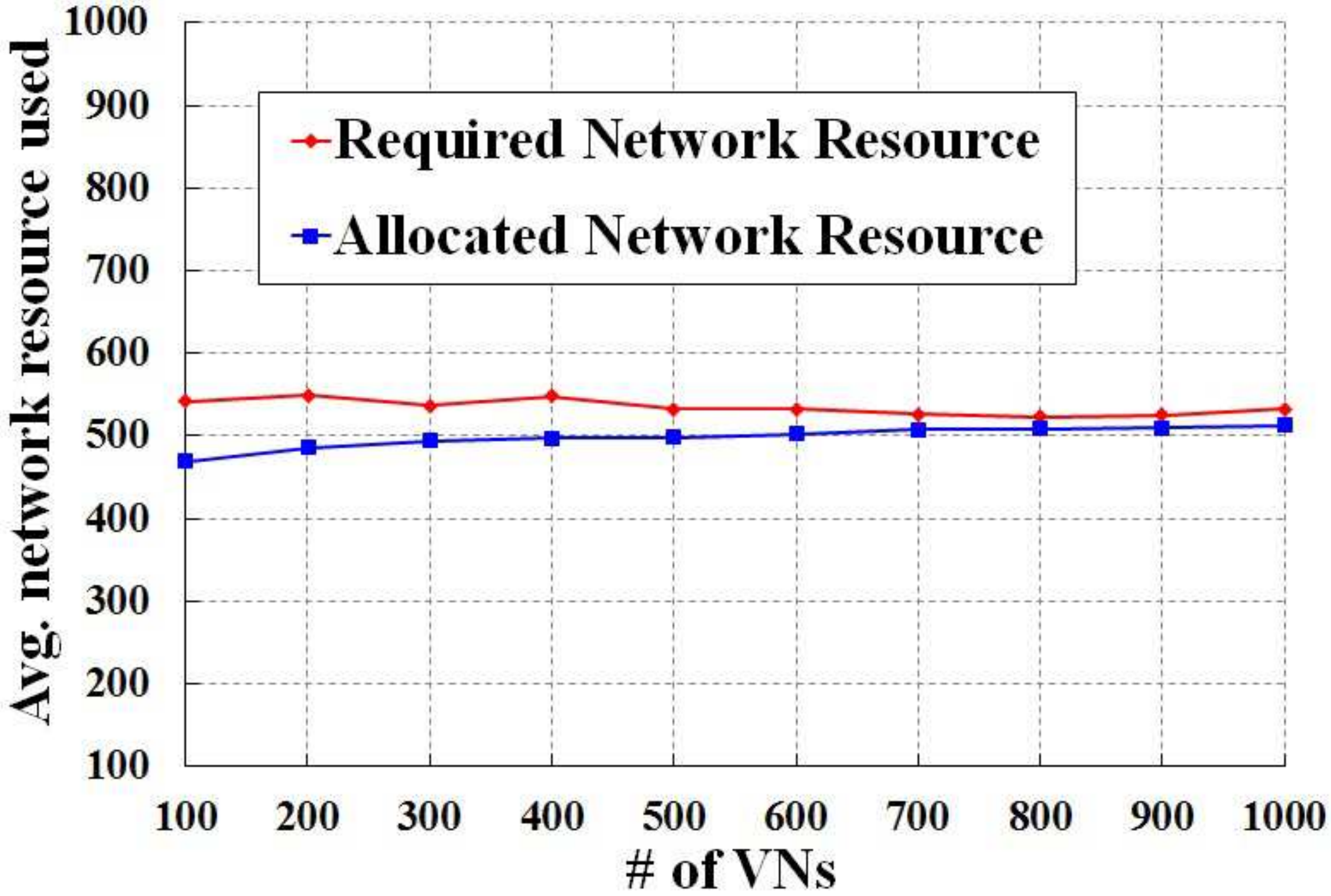,width=65mm}
        \vspace{-2mm}
        \caption{Comparison of required bandwidth and allocated bandwidth.}\vspace{-5mm}
        \label{fig:sim:nw_resrc_used}
    \end{center}
\end{figure}

\begin{figure}[h]
    \begin{center}
        \epsfig{file=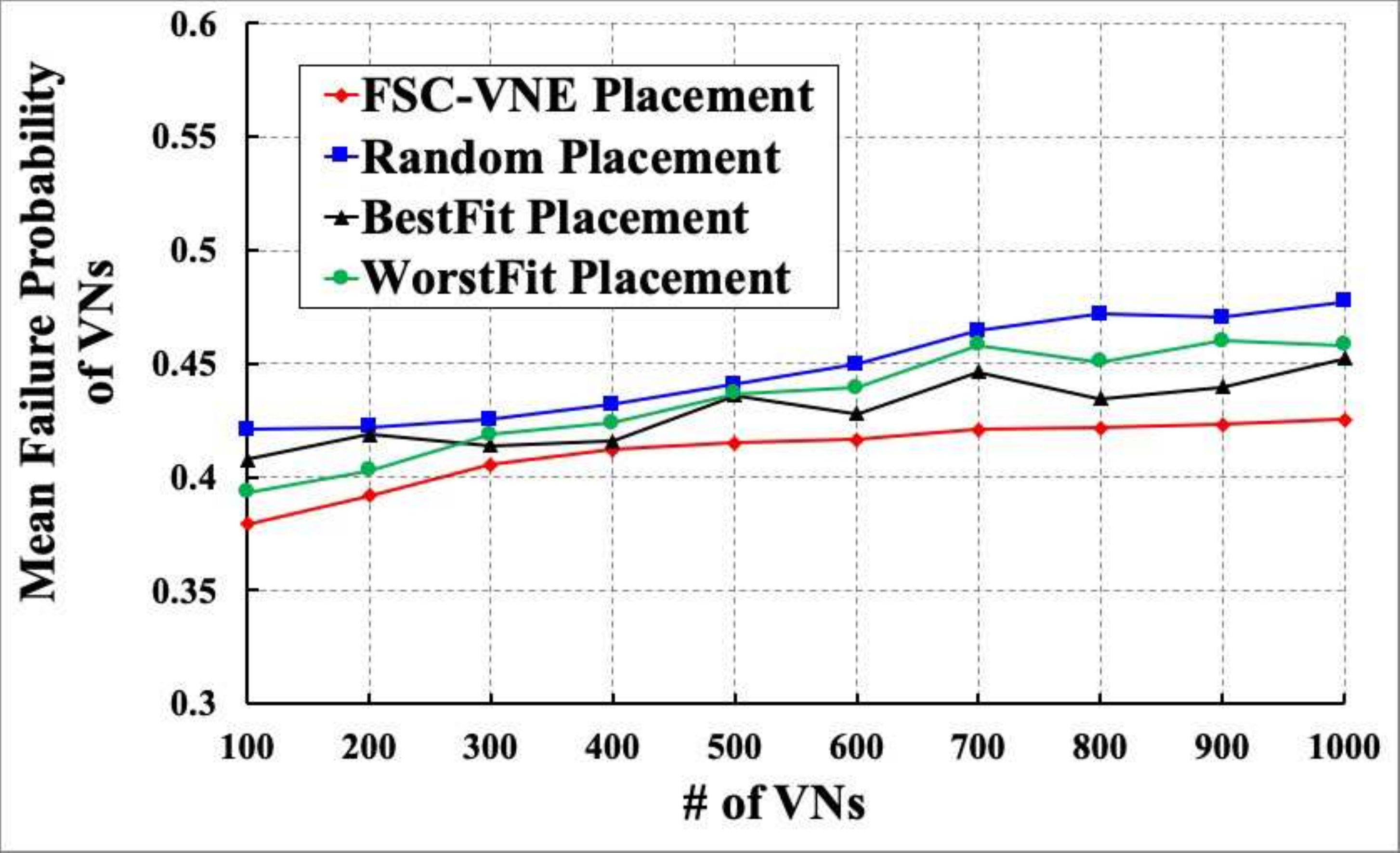,width=65mm}
        \vspace{-2mm}
        \caption{Mean failure probability of VNs.}\vspace{-5mm}
        \label{fig:sim:mean_failProb}
    \end{center}
\end{figure}

Fig. \ref{fig:sim:embedTime} shows the relationship of number of
PMs and number of VN with the embedding time. It is observed that
the  embedding time may increase with increasing number of VNs and
the PMs. The average time required to embed $200$ VNs onto $200$
PMs is $0.005sec$. However, the time increases to approximately
$1sec$ to embed $2000$ VNs onto the same number of PMs, which is
mainly due to the first stage of embedding process as it follows
the centralized approach.

The bandwidth demand of the virtual links can be ignored by
mapping the adjacent VMs onto the single PM. As a result, the
total amount of required bandwidth can be more than the total
amount of allocated bandwidth as shown in Fig.
\ref{fig:sim:nw_resrc_used}. It is observed that the average
network bandwidth demanded by $100$ VNs is $542units$. However,
only $469 units$ of network resource is actually allocated. The
same trend is observed even when the number of VNs is increased.
The average network resource allocated for $1000$ VNs is $513
units$ for the demand of $532 units$ of network resources.

\begin{figure}[ht]
    \begin{center}
        \epsfig{file=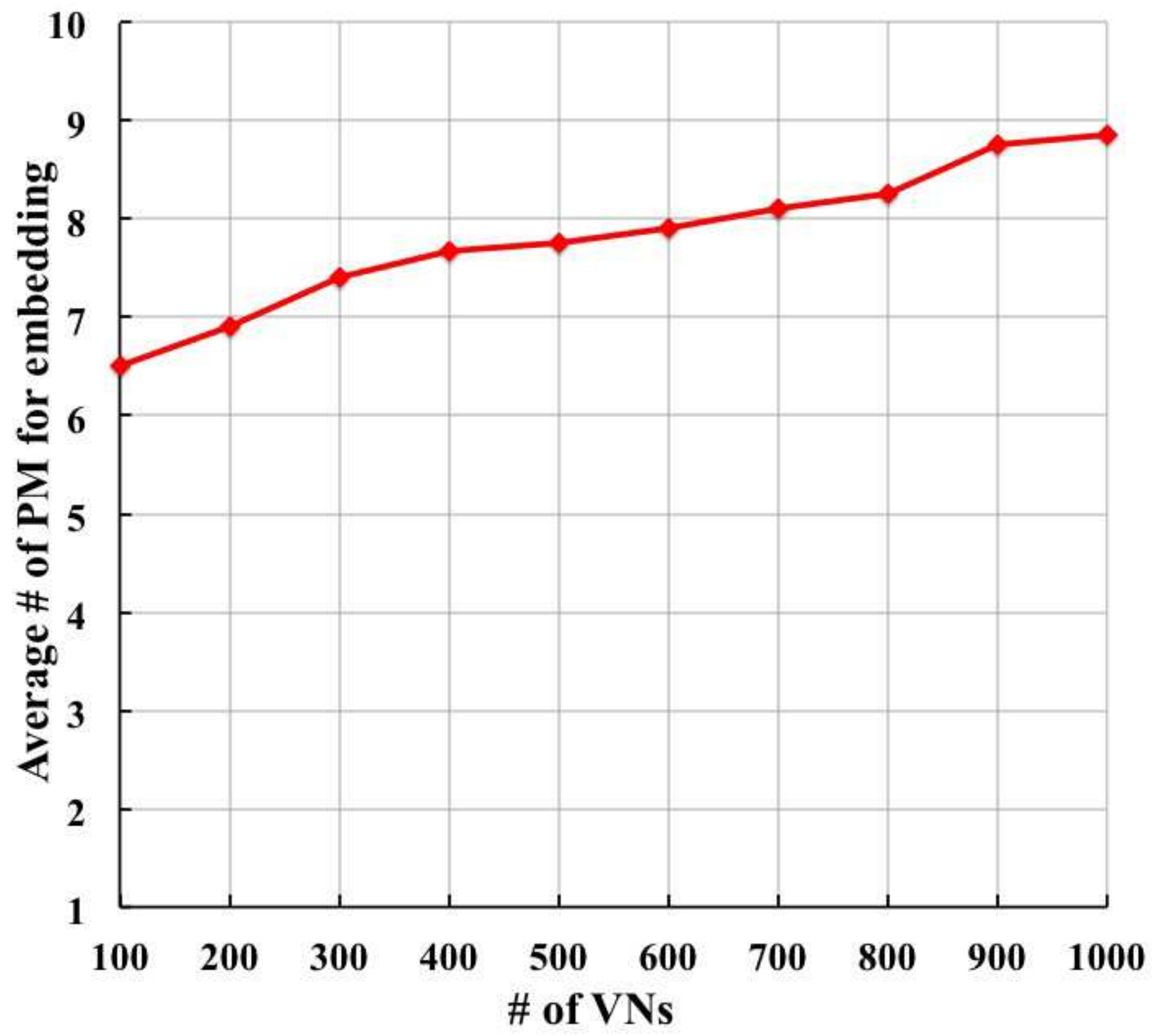,width=60mm}
           \vspace{-3mm}
        \caption{Average number of PMs used for embedding purpose.}\vspace{-5mm}
        \label{fig:sim:embed_resrc}
    \end{center}
\end{figure}

In the second stage of the proposed FSC-VNE algorithm, the PMs are
given higher priority based on the failure probability. The PMs
with lesser failure probability pick-out the VMs first.
Eventually, the VMs are mapped onto the PMs with less failure
probability. The simulation results in Fig. \ref{fig:sim:mean_failProb} demonstrates the comparison of FSC-VNE placement schemes with three other placement strategies: random, BestFit, and WorstFit. In the BestFit placement strategy, PM with smallest available resource picks-out the VMs. Contrary to this, the PM with largest available resource picks-out the VMs in WorstFit strategy in case of FSC-VNE. In random placement, the PMs randomly pick-out the VMs. In all the three placement strategies, the PMs do not consider their failure probability while choosing a VM. Here, the mean failure probability of the VN refers to as the mean of failure probability of the PMs that are hosting the corresponding VNs. For $100$ number of VNs, the mean failure probability is approximately $0.38$ in case of FSC-VNE algorithm. Whereas, the mean failure probability ranges between $0.39$ and $0.42$ for the same number of VNs in all other placement strategies. It is observed that this value increases significantly to $0.41$ in case of FSC-VNE, when the number of VNs increases to $300$. Following this, a gradual increase is observed for further increase in the number of VNs. The mean failure probability increases to approximately $0.42$ when $1000$ number of VNs are embedded. On the other hand, a larger mean failure probability is observed in case of other placement strategies, which ranges between $0.45$ and $0.48$ after $1000$ VNs are embedded. Since those three placement strategies do not consider the failure probability in the embedding process, there mean failure probability values fluctuate, as shown in Fig. \ref{fig:sim:mean_failProb}.

The major disadvantage of the semi-centralized approach is the
resource requirement to embed the VNs. In Fig. \ref{fig:sim:embed_resrc}, the
number of PMs used to embed each VN is presented. It is observed
that the average number of PMs ranges between $6$ and $9$ for $1000$ VNs. This is due to the fact that a single server is responsible for embedding all incoming VNs in centralized approach. However, multiple servers are employed to embed each VN in semi-centralized approach. This disadvantage can be handled by restricting the number of PMs to a lower value. Further, this can also be ignored if a significant improvement in embedding is observed while processing a large number of VNs.

\begin{figure}[ht]
    \begin{center}
        \epsfig{file=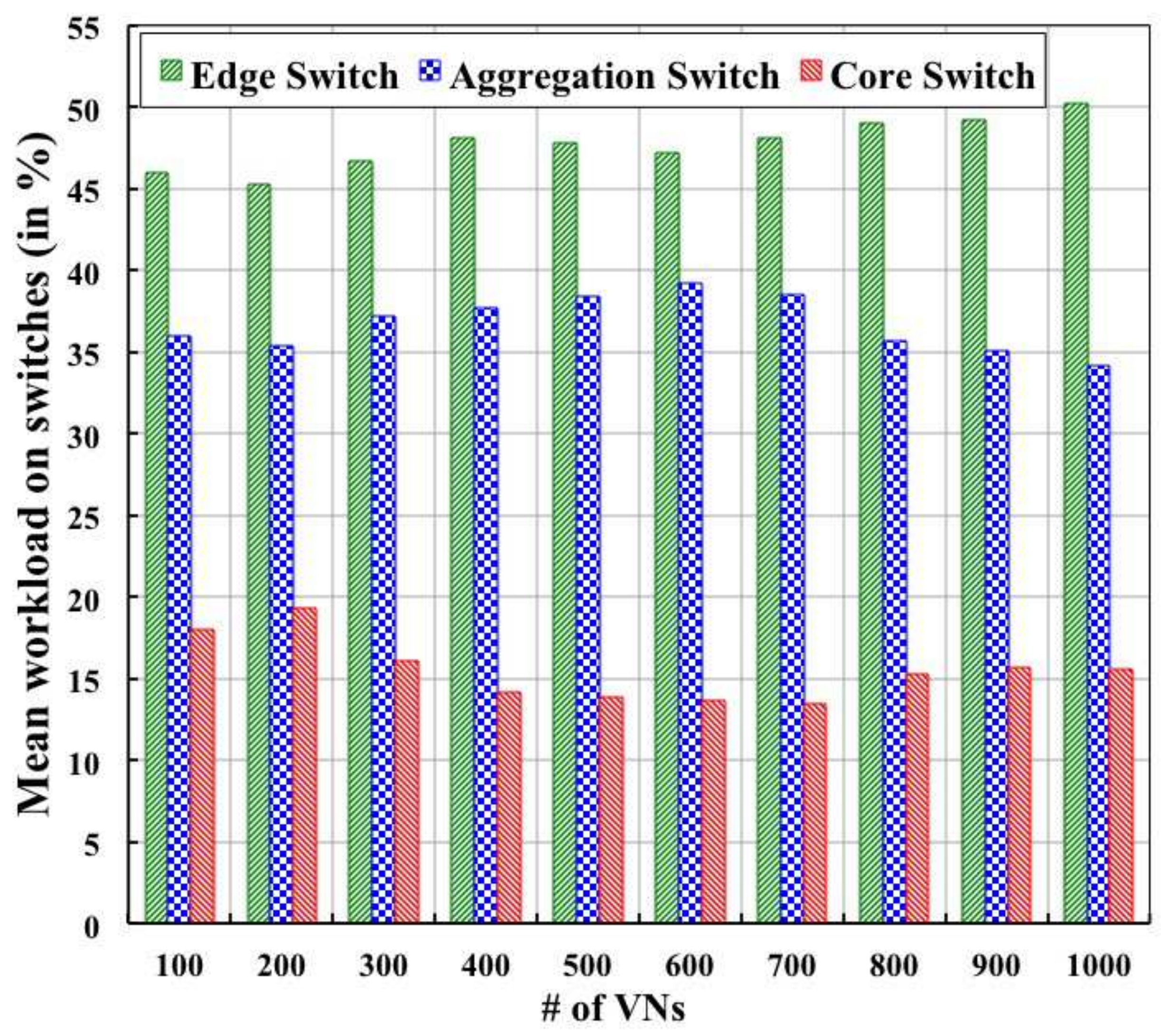,width=60mm}\vspace{-3mm}
        \caption{Workload on different layers of switches.}\vspace{-5mm}
        \label{fig:sim:switch_workload}
    \end{center}
\end{figure}

It is essential to observe the workload on all the three layers of
switches, i.e., edge switches, aggregation switches, and core
switches. In Fig. \ref{fig:sim:switch_workload}, X-axis represents
the number of VNs and Y-axis represents the workload on different
layers of switches. Here, workload refers to as percentage of
communication of a VN being carried out by particular type of
switch. It is observed that for $100$ VNs approximately only
$18\%$, $36\%$, and $46\%$ of the communications among VMs
dependent on core switches, aggregation switches, and edge
switches, respectively. Similarly, an average of $16\%$, $34\%,$ and $50\%$ of the VMs communications for $1000$ VNs dependent on core, aggregation, and edge switches, respectively as shown in Fig. \ref{fig:sim:switch_workload}. The workload on core switches is less than that of the aggregation switches, which is further less than that of the edge switches. This is due to the fact that VMs are embedded onto the PMs with closure distance, which involve lesser number of switches.

\section{Conclusions} \label{sec:conclsn}
In this paper, we have investigated the VNE problem on Fat-Tree
data center network and have presented a novel and effective
Failure-aware Semi-Centralized VN Embedding (FSC-VNE) solution. In
the proposed embedding solution, VNs are mapped onto the existing
physical network in two-staged semi-centralized manner by taking
the advantages of both centralized and distributed embedding
approach. Multiple VMs are allowed to be embedded onto the single
PM. Further, the VMs are mapped onto a set of PMs with minimum
distance. This infers that less number of PMs and switches are
used in an embedding solution. The PMs with minimum failure
probability are given higher preference to host the VMs. This
helps us achieve the goal to minimize the virtual resource
failure probability. Besides, extensive simulations are performed
and our simulation results are compared with other similar
algorithms. The simulation results indicate the superiority and
notable strengths of our FSC-VNE over other algorithms in terms of
acceptance ratio, minimization of required physical resource,
embedding time and VN failure probability. For future work, we
endeavor to implement the proposed FSC-VNE solution in real cloud
environment in order to verify and improve the simulation results.

\section*{Acknowledgment}\label{Ack}
This work was supported in part by the Ministry of Science and
Technology (MOST), Taiwan, under Grant 108-2221-E-182-050 and in
part by Chang Gung Medical Foundation, Taiwan under Grant CMRPD
2J0141.


\bibliographystyle{IEEEtran}
\bibliography{references}
\vspace{-10mm}
\begin{IEEEbiography}[{\includegraphics[width=1in,height=1.25in,clip,keepaspectratio]{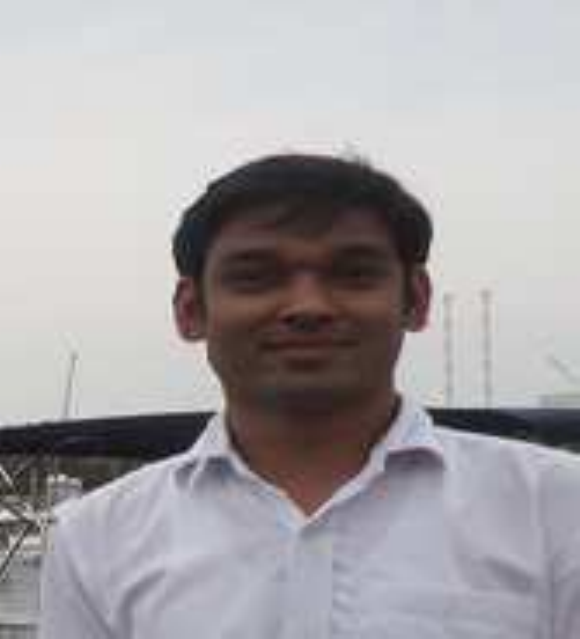}}]{Chinmaya Kumar Dehury}
    Chinmaya Kumar Dehury received bachelor degree from Sambalpur University, India, in June 2009 and MCA degree from Biju Pattnaik University of Technology, India, in June 2013. He received the PhD Degree in the department of Computer Science and Information Engineering, Chang Gung University, Taiwan. Currently, he is a postdoctoral research fellow in the Mobile \& Cloud Lab, Institute of Computer Science, University of Tartu, Estonia. His research interests include scheduling, resource management and fault tolerance problems of Cloud and fog Computing, and the application of artificial intelligence in cloud management. He is also reviewer to several journals and conferences, such as IEEE TPDS, IEEE JSAC, Wiley Software: Practice and Experience, etc.
\end{IEEEbiography}
\vspace{-5mm}
\begin{IEEEbiography}[{\includegraphics[width=1in,height=1.25in,clip,keepaspectratio]{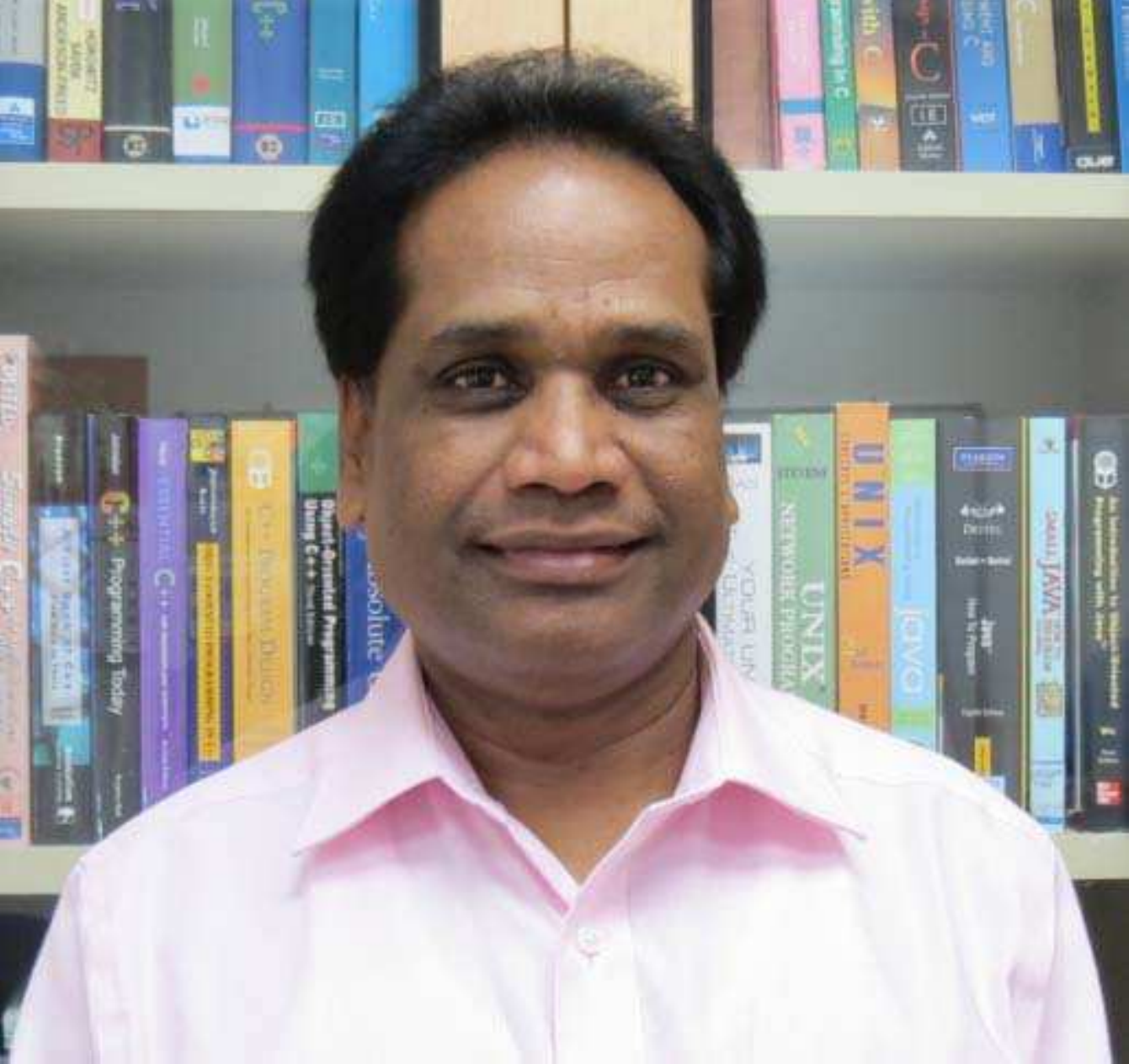}}]{Prasan Kumar Sahoo}
Prasan Kumar Sahoo received the B.Sc. degree in physics (with
Honors), the M.Sc. degree in mathematics both from Utkal
    University, Bhubaneswar, India, in 1987 and 1994, respectively. He
    received the M.Tech. degree in computer science from the Indian
    Institute of Technology (IIT), Kharagpur, India, in 2000, the
    first Ph.D. degree in mathematics from Utkal University, in 2002,
    and the second Ph.D. degree in computer science and information
    engineering from the National Central University, Taiwan, in 2009.
    He is currently a Professor in the Department of Computer Science
    and Information Engineering, Chang Gung University, Taiwan. He is an Adjunct Researcher in the Division of Colon and
    Rectal Surgery, Chang Gung Memorial Hospital, Linkou, Taiwan since 2018. He has worked as an Associate Professor in
    the Department of Information Management, Vanung University,
    Taiwan, from 2007 to 2011. He was a Visiting Associate Professor in the Department of Computer Science, Universite Claude Bernard Lyon 1, Villeurbanne,
    France. His current research interests include artificial
    intelligence, big data analytic, cloud computing, and IoT. He is
    the Lead Guest Editor, special issue of Electronics journal and an Editorial Board Member for the International Journal of Vehicle
    Information and Communication Systems. He has worked as the
    Program Committee Member of several IEEE and ACM conferences and is a senior member, IEEE.
\end{IEEEbiography}

\end{document}